\documentclass{IOS-Book-Article}   

\usepackage[utf8]{inputenc}
\usepackage[english]{babel}
\usepackage{latexsym}
\usepackage{textcomp}
\usepackage{amsmath}
\usepackage{amsthm}
\usepackage{mathptmx}
\usepackage{xspace}
\usepackage{enumitem}
\usepackage[labelfont=bf]{caption}[2004/07/16]
\usepackage{color}
\usepackage{verbatim}
\usepackage{algorithmicx}
\usepackage[ruled]{algorithm}
\usepackage{algpseudocode}

\makeatletter
\def\@listI{\leftmargin\leftmargini
            \parsep 0\p@ \@plus1\p@ \@minus\p@
            \topsep 4\p@ \@plus2\p@ \@minus2\p@
            \itemsep0\p@}
\let\@listi\@listI
\@listi
\def\@listii {\leftmargin\leftmarginii
              \labelwidth\leftmarginii
              \advance\labelwidth-\labelsep
              \parsep 0\p@ \@plus1\p@ \@minus\p@
              \topsep 0\p@ \@plus2\p@ \@minus2\p@}
\def\@listiii{\leftmargin\leftmarginiii
              \labelwidth\leftmarginiii
              \advance\labelwidth-\labelsep
              \parsep 0\p@ \@plus1\p@ \@minus\p@
            	\topsep 0pt \@plus1pt \@minus1pt
            	\itemsep0\p@}
\belowdisplayskip\abovedisplayskip
\makeatother
\newtheorem{thm}{Theorem}

\newtheorem{lem}[thm]{Lemma}
\newtheorem{lemma}[thm]{Lemma}
\newtheorem{prop}[thm]{Proposition}
\newtheorem{defn}[thm]{Definition}
\newtheorem{definition}[thm]{Definition}

\numberwithin{equation}{section}
\newtheorem{Example}{Example}

\newcommand{\To}{\Rightarrow}
\newcommand{\defeater}{\leadsto}

\newcommand{\PROP}{\ensuremath{\mathrm{PROP}}\xspace}
\newcommand{\MOD}{\ensuremath{\mathrm{MOD}}\xspace}
\newcommand{\LIT}{\ensuremath{\mathrm{Lit}}\xspace}
\newcommand{\MODLIT}{\ensuremath{\mathrm{ModLit}}\xspace}

\newcommand{\FACTS}{\ensuremath{F}\xspace}
\newcommand{\LAB}{\ensuremath{\mathrm{Lab}}\xspace}

\newcommand{\non}{\ensuremath{\mathord{\sim}}}
\newcommand{\seq}[2][n]{#2_{1},\dots,#2_{#1}}
\newcommand{\opseq}[4][\otimes]{#2_{#3}#1\cdots#1#2_{#4}}

\newcommand{\bigoslash}{\oslash}
\newcommand{\set}[2][\relax]{%
 \ifthenelse{\equal{#1}{auto}}{\left\{#2\right\}}{#1\{#2#1\}}%
}

\newcommand{\OBL}{\ensuremath{\mathrm{O}}\xspace}
\newcommand{\PERM}{\ensuremath{\mathrm{P}}\xspace}

\newcommand{\red}[1]{{\color{red} #1}}

\let\arrow\hookrightarrow

\begin{document}

\begin{frontmatter}                           

\setlength{\abovedisplayskip}{100pt}
\title{Computing Strong and Weak\\ Permissions in Defeasible Logic%
}
\runningtitle{Computing Defeasible Permissions}

\author[A]{\fnms{Guido} \snm{Governatori}},
\author[A,B,C]{\fnms{Francesco} \snm{Olivieri}}\linebreak
\author[D]{\fnms{Antonino} \snm{Rotolo}} and
\author[A,B,C]{\fnms{Simone} \snm{Scannapieco}}
\address[A]{NICTA, Queensland Research Laboratory, Australia}
\address[B]{Department of Computer Science, University of Verona, Italy}
\address[C]{Institute for Integrated and Intelligent Systems, Griffith University, Australia}
\address[D]{CIRSFID, University of Bologna, Italy}
\runningauthor{Governatori et al.}

\begin{abstract}
\; In this paper we propose an extension of Defeasible Logic to represent and
compute three concepts of defeasible permission. In particular, we discuss
different types of explicit permissive norms that work as exceptions to
opposite obligations. Moreover, we show how strong permissions can be
represented both with, and without introducing a new consequence relation for
inferring conclusions from explicit permissive norms. Finally, we illustrate
how a preference operator applicable to contrary-to-duty obligations can be
combined with a new operator representing ordered sequences of strong
permissions which derogate from prohibitions. The logical system is studied
from a computational standpoint and is shown to have liner computational
complexity.
\end{abstract}

\end{frontmatter}

\pagestyle{plain}

\makeatletter
\abovedisplayskip 5.5pt \@plus 3pt \@minus 3pt
\belowdisplayskip\abovedisplayskip
\abovedisplayshortskip 3pt\@plus 1pt\@minus 1pt
\belowdisplayshortskip\abovedisplayshortskip
\makeatother

\section{Introduction} 
\label{sec:introduction}

The concept of permission plays an important role in many normative domains in
that it may be crucial in characterising notions such as those of
authorisation and derogation \cite{Boella:2005,sartor:2005,stolpe-jal:2010}.
For example, sometimes it may happen that we mistakenly drive to a building
site, or a road-work restricted area, with signs out saying ``No admittance.
Authorised personnel only''. Or consider when we subscribe to an on-line sale
agreement accepting to enter our personal data on the condition that this
information is only used for shipping, and other necessary purposes to
communicate with us or deliver the products to us. In both cases, a permission
(to enter a restricted area or to use our personal data) is stated as an
exception to a general prohibition.

Despite this fact, the concept of permission is still elusive in this field of
literature and has not been extensively investigated in deontic logic as the
notion of obligation. For a long time, deontic logicians mostly viewed
permission as the dual of obligation: $\PERM a \equiv \neg \OBL \neg a$. This
view is unsatisfactory, as it hardly allows us to grasp the meaning of
examples like the ones previously mentioned. For this, and other reasons, the
attempt to reduce permissions to duals of obligations has been criticised (see
\cite{alchourron-bulygin:1984,alchourron:1993}).

One important distinction that has traditionally contributed to a richer
account of this concept is the one between \emph{weak} (or \emph{negative})
and \emph{strong} (or \emph{positive}) permission \cite{vonwright:1963}. The
former corresponds to saying that some $a$ is permitted if $\neg a$ is not
provable as mandatory. In other words, something is allowed by a code iff(only
when) it is not prohibited by that code. At least when dealing with
unconditional obligations, the notion of weak permission is trivially
equivalent to the dual of obligation \cite{makinson-torre:2003}.

The latter concept of strong permission is more complicated, as it
amounts to saying that some $a$ is permitted by a code iff such a
code explicitly states that $a$ is permitted. It follows that a strong
permission is not derived from the absence a prohibition, but is explicitly
formulated in a permissive norm. The complexities of this concept depend on
the fact that, besides ``the items that a code explicitly pronounces to be
permitted, there are others that in some sense follow from the explicit
ones''. The problem is hence ``to clarify the inference from one to the
other'' \cite[p. 391--2]{makinson-torre:2003}. For example, if some $b$
logically follows from $a$, which is strongly permitted, can we say that $b$
is also strongly permitted?

Features such as the distinction between strong and weak permission show the
multi-faceted nature of permission and permissive norms, which has been
overlooked by most logicians for a long time. Nevertheless, some
exceptions have recently offered significant contributions to the logical
understanding of permission
\cite{makinson-torre:2003,boella-torre-icail:2003,boella-torre-nrac:2003,brown:2000,stolpe-jal:2010,stolpe-deon:2010}.
These contributions can be roughly summarised into the following points:

\begin{itemize}
 \item despite some scepticism
   \cite{ross:1968,opalek-wolenski:1991} and critical remarks
   \cite{alchourron-bulygin:1981,alchourron-bulygin:1984} (a discussion of
   this related work can be found in Section \ref{sec:conclusions}), the
   distinction between weak and strong permission seems to be needed,
   otherwise it is rather hard to account for the fact that certain
   permissions make sense because they explicitly derogate to \emph{existing}
   prohibitions while other permissions are not explicit and occur
   \emph{precisely} because opposite prohibitions \emph{do not exist};
\item we may have different types of strong permissions (specifically
permissions that logically follow from explicit permissive norms), according
to whether
\begin{itemize}
\item we statically determine what is actually permitted given what is
obligatory and what is explicitly permitted;
\item we dynamically determine ``the limits on what may be prohibited
without violating static permissions'' \cite{boella-torre-icail:2003};
\end{itemize}
 \item especially in the \emph{law}, strong permissions state exceptions to
   obligations \cite{bobbio:1958}: indeed, derogating with a permission,
   for example, to a general prohibition to use private protected data
   provides an exception to such a prohibition;
\item strong permissions make sense even when any incompatible prohibitions
are not in the legal system; permissions have a dynamic behaviour and prevent
future prohibitions from holding in general, or applying to specific contexts \cite{bulygin:1986}.
\end{itemize} 

This paper moves from the above points with the specific
purpose of studying the different conceptual and computational aspects of weak
and strong permissions. More precisely, the current contribution works in the
following directions:
\begin{description}
	\item[Permissions and defeasibility.] The concept of permission exhibits
strong connections with the idea of defeasibility. Indeed, an example of a
natural way to capture strong permissions acting as exceptions to obligations
is the one where permissions rebut the conclusions of incompatible
prescriptive norms
\cite{makinson-torre:2003,boella-torre-icail:2003,igpl09policy,stolpe-jal:2010}
or undercut them (i.e., challenge an inference rule of an argument supporting
an opposite obligation) \cite{boella-torre-nrac:2003}).

	\item[Permissions and preferences.] Sometimes explicit derogations of
(existing or possible) prohibitions can be ranked according to some preference
orderings. In other words, given any prescriptive norm prohibiting $a$, more
derogations to this norm can be stated and ranked in a certain preference
sequence. This situation may occur in domains such as the law, where for
instance the lawmaker, when imposing duties for citizens, establishes
conditions to lessen the effect of violating such duties to different degrees,
or exempt people to comply with the duties. We will study these mechanisms and see that
ordered sequences of strong permissions, derogating or making exceptions to
prohibitions, have interesting similarities with ordered sequences of
contrary-to-duty obligations \cite{ajl:ctd,clima}. This is a specific novelty
of our contribution, as such sequences regard permissions (i.e., exceptions)
which are not necessarily incompatible with each other.

	\item[Permissions and computation.] If, as we have mentioned, the concept
   of permission is mostly overlooked in literature, being that its computational
   treatment is basically neglected. To the best of our knowledge, no work
   in deontic logic has extensively explored the computational complexity
   of reasoning about different types of permission. Here, we will attempt
   a first analysis of the problem in the context of a modal extension of
   Defeasible Logic \cite{Antoniou_2001}. Modal Defeasible Logic is a
   computationally efficient logical framework able to capture various
   aspects of non-monotonic and modal reasoning, as well as the defeasible
   character of permissive norms, and recently a possible-world semantics for
   it has been proposed \cite{GovernatoriRC12}. We will study how to
   compute weak and strong permissions with and without introducing a new
   non-monotonic consequence relation for permission. The choice of Defeasible Logic is
   motivated by the fact that it is very efficient. Also, the formal
   language of its (non-modal) modulo is simple, thus allowing us to isolate
   the deontic aspects of permissions and investigate their specific
   computational characteristics.
\end{description}

The layout of the paper is as follows.
Section~\ref{sec:three_concepts_of_permission} introduces and informally
discusses three types of defeasible permission in Defeasible Logic. Section~\ref{sec:logics}
presents the technical machinery and states coherency and consistency
results of the proposed extension. In Section~\ref{sec:algos} we develop the
algorithmic means to state what is mandatory and what is permitted in a given
theory, along with the corresponding computational results in
Section~\ref{sec:CompRes}. Section~\ref{sec:discussion} discusses the system
and illustrates how the logical framework presented in Section
\ref{sec:logics} is able to capture the three types of permission.
Section~\ref{sec:conclusions} discusses some related work and provides a
summary of the paper.

\section{Three concepts of permission}\label{sec:three_concepts_of_permission}

This section is meant to offer a brief and gentle introduction to our formal
language and logic, and to discuss three different types of permission and
their relation with the concept of normative defeasibility. Moreover, we
illustrate the idea of preference over permissions that explicitly derogate to
prohibitions.

These aspects will be formally handled in Section~\ref{sec:logics}.
The whole discussion of the computational aspects of permission in Defeasible Logic is postponed to Sections \ref{sec:algos} and \ref{sec:CompRes}.

\subsection{Informal presentation of the logic}

Let us summarise the basic logical intuitions behind our framework.
%
\begin{enumerate}
  \item\label{en:rules} Permissive and prescriptive norms are represented by
  means of defeasible rules, whose conclusions normally follow unless they are
  defeated by contrary evidence. For example, the rule 
\[
\mathit{Order}\To_{\OBL}
  \mathit{Pay}
\]
 says that, if we send a purchase order, then we will be
  defeasibly obliged to pay; the rule 
\[
\mathit{Order}, \mathit{Creditor}
  \To_{\PERM} \neg \mathit{Pay}
\]
states that if we send an order, in general we are not obliged to pay if we
are creditors towards the vendor for the same amount.

\item\label{en:modalities} Rules introduce modalities: if we have the rule
$a\To_{\OBL} b$ and $a$ holds, then we obtain $\OBL b$. That is to say, in the
scenario where conditions described by $a$ hold, the obligation of doing $b$
is active as well. The advantage is that explicitly deriving modal literals
such as $\OBL b$ adds expressive power to the language, since $\OBL b$
may appear in the antecedent of other rules, which can then be
triggered.

\item\label{en:iteration} For the sake of simplicity, modal literals can only
occur in the antecedent of rules.  In other words, we do not admit nested modalities, i.e., rules such as $a
\To_{\OBL} \PERM b$. This is in line with our idea that the
applicability of rules labeled with modality $\Box$ (where $\Box$ can be
$\OBL$ for obligation or $\PERM$ for permission) is the condition for deriving literals
modalised with $\Box$. 

\item The symbols $\OBL$ and $\PERM$ are not simple labels: they are modalities.
$\OBL$ is non-reflexive\footnote{As it is well-known, in a non-reflexive modal
logic $\Box a$ does not imply $a$, where $\Box$ is a modal operator.}:
consequently, we do not have a conflict within the theory when $\neg a$
is the case and we derive that $a$ is mandatory ($\OBL a$); this amounts to having a
violation. The modality $\PERM$
works in such a way that two rules for $\PERM$ supporting $a$ and $\neg a$ do
not clash, but a rule like $\To_{\PERM} b$ attacks a rule such as $\To_{\OBL}
\neg b$ (and vice versa).

\item Like standard Defeasible Logic, our extension is able to establish the relative
  strength of any rule (thus to solve rule conflicts) and has two types of
  attackable rules: defeasible rules and defeaters. Defeaters in Defeasible Logic are a
  special kind of rules: they are used to prevent conclusions but not to
  support them. For example, the defeater
\[
\mathit{SpecialOrder},\mathit{PremiumCustomer}\leadsto_{\OBL}
\neg \mathit{PayBy7Days}
\]
can prevent the derivation of the obligation for premium customers placing
special orders to pay within the deadline of 7 days, but cannot be used to directly derive any conclusion.
\end{enumerate}


\subsection{Permissions and defeasibility}

The above framework, though simple, allows us to express three basic
types of permissions as well as illustrate interesting connections with the idea
of defeasibility.

\paragraph*{Weak permission.} A first way to define permissions in Defeasible Logic is by
simply considering weak permissions and stating that the opposite of what is
permitted is not provable as obligatory. Let us consider a normative system
consisting of the following two rules:
\[
\begin{array}{llll}
r_1:& \mathit{Park}, \mathit{Vehicle} \To_{\OBL} \neg \mathit{Enter}\\
r_2: &  \mathit{Park}, \mathit{Emergency} \To_{\OBL} \mathit{Enter}.
\end{array}
\]
Here the normative system does not contain any permissive norm. However, since
Defeasible Logic is a sceptical non-monotonic logic, in case both $r_1$ and $r_2$ fire
we neither conclude that it is prohibited nor that it
is obligatory to enter, because we do not know which rule is stronger. Hence, in this
context, both $\neg \mathit{Enter}$ and $\mathit{Enter}$ are weakly permitted.

%

As already argued, this is the most direct way to define the idea of weak
permission: some $q$ is permitted by a code iff $q$ is not prohibited by that
code. Accordingly, saying that any literal $q$ is weakly permitted corresponds to
the failure of deriving $\neg q$ using rules for $\OBL$. Notice that, in Defeasible Logic, this does not amount to obtain $\neg \OBL \neg q$.

\paragraph*{Explicit permissions are defeaters.} In Defeasible Logic any rule can be used to
prevent the derivation of a conclusion. For instance, suppose there exists a
norm that prohibits to U-turn at traffic lights unless there is a ``U-turn
permitted'' sign:
\[
\begin{array}{lllll}
  r_1: \mathit{AtTrafficLight} \To_{\OBL} \neg \mathit{Uturn}\\
  r_2: \mathit{AtTrafficLight, UturnSign} \To_{\OBL} \mathit{Uturn}.
\end{array}
\]
In this example we use a defeasible rule for obligation to block the prohibition to
U-turn. However, this is not satisfactory: if we do not know whether $r_2$ is
stronger than $r_1$, then the best we can say is that U-turn is weakly
permitted. Furthermore, if $r_2$ prevails over $r_1$, we derive that U-turn is
obligatory.

Thus, there are good reasons to argue that defeaters for $\OBL$ are suitable
to express an idea of strong permission\footnote{The idea of using defeaters
to introduce permissions was introduced in \cite{icail05}.}. Explicit rules
such as $r:a\defeater_{\OBL} q$ state that $a$ is a specific reason for
blocking the derivation of $\OBL\neg q$ (but not for proving $\OBL q$). In
other words, this rule does not support any conclusion, but states that $\neg q$ is deontically undesirable. Consider this
example:
\[
\begin{array}{llll}
r_{1}:& \mathit{Weekend},\mathit{AirPollution}
    \To_{\OBL} \neg \mathit{UseCar}\\
r_{2}:& \mathit{Weekend}, \mathit{Emergency} \defeater_{\OBL}
\mathit{UseCar}.
\end{array}
\]
Rule $r_{1}$ states that on weekends it is forbidden to use private cars if a
certain air pollution level is exceeded. Defeater $r_{2}$ is in fact an
exception to $r_1$, and so it seems to capture the above idea that
explicit permissive norms (especially in the law) provide exceptions to
obligations.
%
%
%
%
\paragraph*{Explicit permissions using permissive rules.} Another approach is based on introducing
specific rules for deriving permissions
\cite{makinson-torre:2003,boella-torre-icail:2003}.
Let us consider the following situation:
\[
\begin{array}{llll}
r_{1}:& \mathit{Weekend},\mathit{AirPollution}
    \To_{\OBL} \neg \mathit{UseCar} \\
r'_{2}:& \mathit{Emergency} \To_{\PERM} \mathit{UseCar}.
\end{array}
\]
As $r_2$ in the previous scenario, $r'_2$ looks like an exception to $r_1$.
The apparent difference between $r_2$ and $r'_2$ is that the latter is
directly used to prove that the use of the car is permitted
($\PERM\mathit{UseCar}$) in case of emergencies. The question is:
does it amount to a real difference?

Although $r_2$ is a defeater, it is specifically used to derive the strong
permission to use the car, like $r'_{2}$. In addition, rules such as $r'_2$ do
not attack other permissive rules, but are in conflict only with rules for
obligation intended to prove the opposite conclusion. This precisely holds for
defeaters.

Moreover, let us suppose to have the defeater $s:a\defeater_{\PERM} b$. Does $s$
attack a rule like $\To_{\PERM} \neg b$? 

If this is the case, $s$ would be close to an obligation. The fact that
$\PERM b$ does not attack $\PERM \neg b$ makes it pointless for $s$ to
introduce defeaters for $\PERM$. But, if this is not the case, $s$ could only
attack $\To_{\OBL} \neg b$, thus being equivalent to $s':a\defeater_{\OBL} b$.

Therefore, although it is admissible to have defeaters,
we do not need to distinguish defeaters for $\OBL$ from those for $\PERM$. One
way to mark the difference between $\defeater$ and $\To_{\PERM}$ is by
stating that only the latter rule type admits ordered sequences of strong
permissions in the head of a rule, which are supposed to derogate or make
exceptions to prohibitions. This matter will be discussed in the next
subsection.

\subsection{Permissions, obligations, and preferences}

The introduction of ordered sequences of strong
permissions in the head of a rule, which derogate or make
exceptions to prohibitions, can be logically modelled by enriching the formal
language and following these guidelines:

\begin{enumerate}
\item\label{en:otimes} In many domains, such as the law, norms often specify
mandatory actions to be taken in case of their violation. In general,
obligations in force after the violation of some other obligations correspond
to contrary-to-duty (CTD) obligations. These constructions affect the formal
characterisation of compliance since they identify situations that are not
ideal, but still acceptable. A compact representation of CTDs may resort to
the non-boolean connective $\otimes$ \cite{ajl:ctd}: a formula like
$x\To_{\OBL} a\otimes b$ means that, if $x$ is the case, then $a$ is
obligatory, but if the obligation $a$ is not fulfilled, then the obligation
$b$ is activated and becomes in force until it is satisfied, or violated.

\item Concepts introduced at point \ref{en:otimes} can be extended to
permissive rules with the subscripted arrow $\To_{\PERM}$ by introducing the
non-boolean connective $\odot$ for sequences of permissions. As in the case of
$\otimes$, given a rule like $\To_{\PERM} a \odot b$, we can proceed through
the \mbox{$\odot$-chain} to obtain the derivation of $\PERM b$. However,
permissions cannot be violated, and consequently it does not make sense to
obtain $\PERM b$ from $\To_{\PERM} a \odot b$ and $\neg a$. In this case, the
reason to proceed in the chain is rather that the normative system allows us
to prove $\OBL \neg a$. Hence, $\odot$ still establishes a preference
order among strong permissions and, in case the opposite obligation is in
force, another permission holds. This is significant especially when strong
permissions are exceptions to obligations.
\end{enumerate}

In this paper we take a neutral approach as to whether ordered sequences of
obligations or permissions are either given explicitly, or
inferred from other rules. However, we point out that
normative documents often explicitly contains provision with such structures.
A clear example of this is provided by the Australian ``National Consumer
Credit Protection Act 2009'' (Act No.~134 of 2009) which is structured in
such a way that for every section establishing an obligation or a prohibition,
the penalties for violating the provision are given in the section itself.

\begin{Example}[National Consumer Credit Protection Act 2009] Section 29
(Prohibition on engaging in credit activities without a licence) of the act recites:
\begin{quote}
  (1) A person must not engage in a credit activity if the person does not
  hold a licence authorising the person to engage in the credit activity.\\
  Civil penalty: 2,000 penalty units.\\
  {}[\dots]\\
  Criminal penalty: 200 penalty units, or 2 years imprisonment, or both.
\end{quote}
This norm can be represented as
\begin{gather*}
  r_{1}: \Rightarrow_{\OBL} 
  \neg \mathit{CreditActivity}
  \otimes 
  \mathit{2000CivilPenaltyUnits}\\
  r_{2}: \mathit{CreditLicence} \Rightarrow_{\PERM}
  \mathit{CreditActivity}
\end{gather*}
where $r_{2}>r_{1}$. The first rules state that in absence of other
information a person is forbidden to engage in credit activities
($\OBL\neg\mathit{CreditActivity}$), and then the second rule establish an
exception to the prohibition, or in other terms it recites a condition under
which such activities are permitted. The section then continues by giving
explicit exceptions (permissions) to the prohibition to engage in credit
activity, even without a valid licence.
\end{Example}
Sequences of permissions are a natural fit for expressions like ``the subject
is authorised, in order of preference, to do the following: (list)'' or ``the
subject is entitled, in order of preference, to one of the following:
(list)''. This is illustrated in the next example.

\begin{Example}[U.S. Copyright Act]
A concrete instance of sequences of permissions is given by Section 504(c)(1) (Remedies for infringement:
Damages and profits) of the U.S. Copyright Act (17 USC \S 504).
\begin{quote}
  Except as provided by clause (2) of this subsection, the copyright
  owner may elect, at any time before final judgment is rendered, to
  recover, instead of actual damages and profits, an award of statutory
  damages for all infringements involved in the action, with respect to any
  one work, for which any one infringer is liable individually, or for which
  any two or more infringers are liable jointly and severally, in a sum of
  not less than \$750 or more than \$30,000 as the court considers just.
  [\dots]
\end{quote}
The above provision can be modelled as 
\[
  \mathit{infringment}, \mathit{beforeJudgment} \Rightarrow_{\PERM}
  \mathit{ActualDamages} \odot \mathit{StatutoryDamages}
\]
The above rendering of the textual provision is based on the interpretation of
the term `instead', which suggests that the copyright owners are entitled by
default the award of the actual damages and profits, but they may elect to
recover statutory damages, which is then the second option if exercised by the
relevant party.\footnote{Here we speak of \emph{entitlements} or
\emph{rights}. A \emph{right} is a permission on one party (in this case the
copyright owner) generating an obligation on another party (in this case the
infringer). For a more detailed discussion on the concept of right see
\cite{sartor:2005}.}
\end{Example}
As we have just seen, chains of obligations are appropriate to capture the
obligations and the penalties related to them. Furthermore, this kind of
structure has been successfully used for applications in the area of business
process compliance \cite{ruleml12demo}.
In a situation governed by the rule $\Rightarrow_{\OBL} a \otimes b$ and where
$\neg a$ and $b$ hold, the norm has been complied with (even if to a
lower degree than if we had $a$). On the contrary, if we had two rules
$\Rightarrow_{\OBL} a$ and $\neg a\Rightarrow_{\OBL}b$, then the first norm
would have been violated, while the second would have been complied with. But
in overall, the whole case would be not compliant
\cite{Governatori_Sadiq_2008}.

Consider the following example: 
\[
\begin{array}{ll}
r_{1}: & \mathit{Invoice} \Rightarrow_{\OBL}  \mathit{PayWithin7days} \\
r_{2}: & \OBL  \mathit{PayWithin7days}, \neg \mathit{PayWithin7days} \Rightarrow_{\OBL} \mathit{Pay5\%Interest} \\
r_{3}: & \OBL \mathit{Pay5\%Interest}, \neg \mathit{Pay5\%Interest} \Rightarrow_{\OBL} \mathit{Pay10\%Interest}.
\end{array}
\]
What happens if a customer violates both the obligation to pay within 7 days
after the invoice and the obligation to pay the 5\% of interest, but she pays
the total amount plus the 10\% of interest? In the legal perspective the
customer should be still compliant, but in this representation contract
clauses $r_{1}$ and $r_{2}$ have been violated. However, if we represent the
whole scenario with the single rule
\[
  \mathit{Invoice} \Rightarrow_{\OBL} 
    \mathit{PayBy7days} \otimes
    \mathit{Pay5\%Interest} \otimes
    \mathit{Pay10\%Interest},
\]
then the rule is not violated, and the customer is compliant with the contract.

Even when the text of legal provisions does not explicitly have this
form, there are cases where the joint interpretation of several legal
provisions still leads to formulate applicable norms with orders among
derogations.

\begin{Example}[Formal equality and affirmative action]
Art. 3, 1st paragraph, of the Italian constitution ensures formal equality of
citizens (in fact, all individuals) before the law, namely, an equal legal
treatment for everybody:

\begin{quote}
All citizens have equal social dignity and are equal before the law, without
distinction of sex, race, language, religion, political opinion, personal and
social conditions. [\dots]
\end{quote}
This general principle can be sometimes derogated, for example, when
derogations are meant ``to remove those obstacles of an economic or social
nature which constrain the freedom and equality of citizens'' (art. 3, 2nd
par.). In fact, one may argue that permitting (which is different from
imposing as mandatory) the adoption of affirmative action policies in favour
of women (e.g., introducing quotas for women in politics and the job market)
is a flexible legal measure to remove some of those obstacles. Now, suppose a
quota for women is guaranteed in public institutions in hiring and promoting
employees, but another similar derogation can be applied to disabled people.
Imagine that, in a specific case, it is not possible to apply both derogations
(for example, this would lead to exceeding the number of jobs available) and
so we have to choose to hire a woman or a disabled man. In absence of any
further legal provision, one possible solution is to balance both options with
respect to the specific facts $X$ of the case, thus ranking, in a rule $r$,
these options in order of preference, given the facts $X$ (on balancing, see
\cite{Alexy:2003,Sartor:2010}). For instance, if disabled men should be
favoured over non-disabled women (because disability in this case reinforces a
more serious discrimination or disadvantage) then $r$ is the following:
\[
r: X \To_{\PERM} \mathit{Hire\_Disabled\_Men} \odot \mathit{Hire\_NonDisabled\_Women}
\]
The reason why we should still keep as a second option
$\mathit{Hire\_NonDisabled\_Women}$ depends on the fact that we can draw only
defeasibly the permission of $\mathit{Hire\_Disabled\_Men}$. Indeed, we have
only considered art. 3 of the Italian constitution but other legal provisions
or factual reasons could block this conclusion. For example, suppose that the
disabled person applying for the job was some years earlier convicted of the
crime of belonging to a mafia organisation, while the law prohibits in general
and without exceptions for public institutions to hire people who committed
that crime. Or imagine that, in the meantime, the disabled man has withdrawn
his request for a job. In both cases, despite $X$ occurs, the first option
does not hold and, all things considered, the second one can be applied in
order to derogate to art. 3, 1st par., of the Italian constitution.
\end{Example}

\section{Defeasible Deontic Logic with strong permission} \label{sec:logics}

This section begins by introducing the language adopted to formalise
obligations and strong permissions in DL, and describing the inferential
mechanism in the form of proof conditions defining the logic. Finally, we show
that the proposed formalisation enjoys properties appropriate to model the
notion of strong permission. We will proceed incrementally: this section, as
well as Section \ref{sec:algos}, works only with obligations and
strong permissions expressed by rules for $\PERM$. In
Section~\ref{sec:discussion} we will show how weak permissions and strong
permissions based on defeaters can be easily captured in the framework.

\medskip

\noindent We consider a logic whose language is defined as follows.

\begin{defn}[Language]
	Let \PROP be a set of propositional atoms, $\MOD = \{\OBL, \PERM\}$ the
set of modal operators {where $\OBL$ is the modality for the obligation and $\PERM$ for permission}.
\begin{itemize}
	\item The set $\LIT=\PROP\cup \{ \neg p\,|\,p\in\PROP\}$ denotes the set of \emph{literals}.
	\item The \emph{complementary} of a literal $q$ is denoted by $\non q$; if $q$ is a positive literal $p$, then $\non q$ is $\neg p$, and if $q$ is a negative literal $\neg p$, then $\non q$ is $p$.
	\item The set of \emph{modal literals} is $\MODLIT=\{\Box l, \neg \Box l\,|\,l\in \LIT,\ \Box\in \MOD\}$.
\end{itemize}
\end{defn}
 
\noindent We introduce two preference operators, $\otimes$ for obligations and
$\odot$ for permissions, and we will use $\oslash$ when we refer to one of
them generically. These operators are used to build chains of preferences,
called $\oslash$-expressions. The formation rules for well-formed
$\oslash$-expressions are:
\begin{enumerate}[label=(\alph*)]
  \item every literal $l \in \LIT$ is an $\oslash$-expression;
  \item if $A$ is an $\otimes$-expression, $B$ is an $\odot$-expression and 
    $c_1, \dots, c_k \in \LIT$, then $A \otimes c_1 \otimes \dots \otimes c_k$ 
    is an $\otimes$-expression, $B \odot c_1 \odot \dots \odot c_k$ is an 
    $\odot$-expression, $A \odot B$ is an $\oslash$-expression;
  \item every $\otimes$-expression and $\odot$-expression is an 
    $\oslash$-expression;
  \item nothing else is an $\oslash$-expression.
\end{enumerate}

\noindent In addition we stipulate that $\otimes$ and $\odot$ obey the
following properties:

\begin{enumerate}
  \item $a\oslash (b \oslash c) = (a \oslash b) \oslash c$ (associativity); 
  \item $\bigoslash_{i=1}^{n} a_i = (\bigoslash_{i=1}^{k-1}a_i) \oslash
(\bigoslash_{i=k+1}^{n\phantom{k}}a_i)$ where there exists $j$ such that $a_j= a_k$
and $j<k$ (duplication and contraction on the right).  
\end{enumerate}
\noindent Given an $\oslash$-expression $A$, the \emph{length} of $A$ is the
number of literals in it. Given an $\oslash$-expression $A\oslash b\oslash C$
(where $A$ and $C$ can be empty), the \emph{index} of $b$ is the length of $A
\oslash b$. We also say that $b$ appears at index $n$ in $A \oslash b$ if the
length of $A \oslash b$ is $n$.

We adopt the standard Defeasible Logic definitions of \emph{strict rules},
\emph{defeasible rules}, and \emph{defeaters} \cite{Antoniou_2001}. However
for the sake of simplicity, and to better focus on the non-monotonic aspects
that Defeasible Logic offers, in the remainder we use only defeasible rules and defeaters.
In addition, we have to take the modal operators into account.

\begin{defn}[Rules]
	Let \LAB be a set of arbitrary labels. Every rule is of the type
\begin{displaymath}
r: A(r) \arrow C(r)
\end{displaymath}
where
\begin{enumerate}
\item $r \in \LAB$ is the name of the rule;
\item $A(r) = \set{\seq{a}}$, the \emph{antecedent} (or \emph{body}) of the
rule, is the set of the premises of the rule (alternatively, it can be
understood as the conjunction of all the literals in it). Each $a_i$ is either
a literal, or a modal literal;
\item $\arrow\in\{\Rightarrow_{\Box}, \leadsto\}$ denotes the type of the rule.
If $\arrow$ is $\Rightarrow_{\Box}$, the rule is a \emph{defeasible rule}, while
if $\arrow$ is $\leadsto$, the rule is a \emph{defeater}. The subscript $\Box \in
\MOD$ in defeasible rules represents the modality introduced by the
rule itself: the mode of a rule tells us what kind of conclusion we obtain
from the rule. As we argued in Section \ref{sec:three_concepts_of_permission}, we do not need to label $\defeater$ with any modality;
\item $C(r)$ is the \emph{consequent} (or \emph{head}) of the rule, which is
an $\oslash$-expression. Two constraints apply on the consequent of a rule:
(a) if $\arrow$ is $\leadsto$, then $C(r)$ is a single literal; (b) if $\Box =
\PERM$, then $C(r)$ must be an $\odot$-expression.
\end{enumerate}
\end{defn}

\noindent Given a set of rules $R$, we will use the following abbreviations
for specific subsets of rules:
\begin{itemize}
\item $R_{def}$ denotes the set of all defeaters in the set $R$;
\item $R[q,n]$ is the set of rules where $q$ appears at index $n$ in the
consequent. The set of (defeasible) rules where $q$ appears at any index $n$ is denoted by $R[q]$;
\item $R^{\Box}$ with $\Box \in \MOD$ denotes the set of all rules in $R$ introducing modality $\Box$;
\item $R^{\OBL}[q,n]$ is the set of (defeasible) rules where $q$ appears
at index $n$ and the operator preceding it is $\otimes$ for $n>1$ or the mode
of the rule is \OBL for $n=1$. The set of
(defeasible) rules where $q$ appears at any index $n$ satisfying the above
constraints is denoted by $R^{\OBL}[q]$;
\item similarly $R^{\PERM}[q,n]$ is the set of rules where $q$ appears at
index $n$, and the operator preceding it is $\odot$ for $n>1$ or the mode of the rule is \PERM for $n=1$. The set of (defeasible)
rules where $q$ appears at any index $n$ satisfying the above constraints is
denoted by $R^{\PERM}[q]$.
\end{itemize}


\begin{defn}
A \emph{Defeasible Theory} is a structure $D = (F,R,>)$, where $F$, the set of
facts, is a set of literals and modal literals, $R$ is a set of rules and $>$,
the superiority relation, is a binary relation over $R$.
\end{defn}

\noindent A theory corresponds to a normative system, i.e., a set of norms
which are modelled by rules. The superiority relation is used for conflicting
rules, i.e., rules whose conclusions are complementary literals, in case both
rules fire. Notice that we do not impose any restriction on the superiority
relation: it is just a binary relation determining the relative strength of
two rules.



\begin{defn}
A \emph{proof} P in a defeasible theory $D$ is a linear sequence $P(1)\dots
P(n)$ of \emph{tagged literals} in the form of $+\partial_{\Box} q$ and
$-\partial_{\Box} q$ with $\Box \in \MOD$, where $P(1)\dots P(n)$ satisfy the
proof conditions given in Definitions~\ref{def:+pO}--\ref{def:-pP}.
\end{defn}

\noindent The tagged literal $+\partial_{\Box } q$ means that $q$ is
\emph{defeasibly provable} in $D$ with modality $\Box$, while
$-\partial_{\Box} q$ means that $q$ is \emph{defeasibly refuted} with modality
$\Box$. The initial part of length $i$ of a proof $P$ is denoted by $P(1..i)$.

The first thing to do is to define when a rule is applicable or discarded. A
rule is \emph{applicable} for a literal $q$ if $q$ occurs in the head of the
rule, all non-modal literals in the antecedent are given as facts and all the
modal literals have been defeasibly proved (with the appropriate modalities).
On the other hand, a rule is \emph{discarded} if at least one of the modal
literals in the antecedent has not been proved (or is not a fact in the case
of non-modal literals). However, as literal $q$ might not
appear as the first element in an $\oslash$-expression in the head of the
rule, some additional conditions on the consequent of rules must be satisfied.
Defining when a rule is applicable or discarded is essential to characterise
the notion of provability for obligations ($\pm\partial_{\OBL}$) and
permissions ($\pm\partial_{\PERM}$).

\begin{defn}\label{defn:APPL+pO}
A rule $r \in R[q,j]$ such that $C(r) = c_{1} \otimes \cdots \otimes c_{l-1}
\odot c_{l} \odot \dots \odot c_{n}$ is \emph{applicable} for literal $q$ at
index $j$, with $1 \leq j < l$, in the condition for $\pm\partial_{\OBL}$ iff

\smallskip
\noindent
\begin{minipage}{\textwidth}
\begin{tabbing}
  $(1)$ \= for all $a_i \in A(r)$: \\
		\> $(1.1)$ if $a_i = \OBL l$ then $+\partial_{\OBL} l \in P(1..n)$; \\
		\> $(1.2)$ if $a_i = \neg \OBL l$ then $-\partial_{\OBL} l \in P(1..n)$; \\
		\> $(1.3)$ if $a_i = \PERM l$ then $+\partial_{\PERM} l \in P(1..n)$; \\
		\> $(1.4)$ if $a_i = \neg \PERM l$ then $-\partial_{\PERM} l \in P(1..n)$; \\
		\> $(1.5)$ if $a_i = l \in \LIT$ then $l \in F$, and \\
   $(2)$ for all $c_k \in C(r)$, $1 \leq k < j$, $+\partial_{\OBL}c_k\in P(1..n)$
   and $ (c_k \not\in F$ or $\non c_k \in F)$.
\end{tabbing}
\end{minipage}
\end{defn}

\smallskip

\noindent Conditions (1.1)--(1.5) represent the requirements on the antecedent
as informally described above; condition (2) on the head of the rule states
that each element $c_k$ prior to $q$ must be derived as an obligation, and a
violation of such obligation has occurred.

\begin{defn}\label{defn:APPL+pP}
A rule $r \in R[q,j]$ such that $C(r) = c_{1} \otimes \cdots \otimes c_{l-1} \odot c_{l} \odot \dots \odot c_{n}$ is \emph{applicable} for literal $q$ at index $j$, with $l \leq j \leq n$ in the condition for $\pm\partial_{\PERM}$ iff

\smallskip
\noindent
\begin{minipage}{\textwidth}
\begin{tabbing}
  $(1)$ \=for all $a_i \in A(r)$: \\
		\> $(1.1)$ if $a_i = \OBL l$ then $+\partial_{\OBL} l \in P(1..n)$; \\
		\> $(1.2)$ if $a_i = \neg \OBL l$ then $-\partial_{\OBL} l \in P(1..n)$; \\
		\> $(1.3)$ if $a_i = \PERM l$ then $+\partial_{\PERM} l \in P(1..n)$; \\
		\> $(1.4)$ if $a_i = \neg \PERM l$ then $-\partial_{\PERM} l \in P(1..n)$; \\
		\> $(1.5)$ if $a_i = l \in \LIT$ then $l \in F$, and \\
   $(2)$ for all $c_k \in C(r)$, $1 \leq k < l$, $+\partial_{\OBL}c_k\in P(1..n)$
   and ($c_k \not\in F$ or $\non c_k \in F$), and\\
   $(3)$ for all $c_k \in C(r)$, $l \leq k < j$, $-\partial_{\PERM} c_k\in P(1..n)$.
\end{tabbing}
\end{minipage}
\end{defn}

\smallskip


\noindent The only difference with respect to $\pm\partial_{\OBL}$ is the
presence of an additional condition, stating that all permissions prior to $q$
must be refuted (condition (3)).


\begin{defn}\label{defn:DISC+pO+pP}
A rule $r \in R[q,j]$ such that $C(r) = c_{1} \otimes \cdots \otimes c_{l-1}
\odot c_{l} \odot \dots \odot c_{n}$ is \emph{discarded} for literal $q$ at
index $j$, with $1 \leq j \leq n$ in the condition for $\pm\partial_{\OBL}$ or
$\pm\partial_{\PERM}$ iff

\smallskip
\noindent
\begin{minipage}{\textwidth}
\begin{tabbing}
  $(1)$ \= there exists $a_i \in A(r)$ such that\\
		\> $(1.1)$ if $a_i = \OBL l$ then $-\partial_{\OBL} l \in P(1..n)$; \\
		\> $(1.2)$ if $a_i = \neg \OBL l$ then $+\partial_{\OBL} l \in P(1..n)$; \\
		\> $(1.3)$ if $a_i = \PERM l$ then $-\partial_{\PERM} l \in P(1..n)$; \\
		\> $(1.4)$ if $a_i = \neg \PERM l$ then $+\partial_{\PERM} l \in P(1..n)$; \\
		\> $(1.5)$ if $a_i = l \in \LIT$ then $l \not\in F$, or \\
   $(2)$ there exists $c_k \in C(r)$, $1 \leq k < l$, such that either
   $-\partial_{\OBL} c_k\in P(1..n)$ or $c_k \in F$, or\\
   $(3)$ there exists $c_k \in C(r)$, $l \leq k < j$, such that $+\partial_{\PERM} c_k\in P(1..n)$.
\end{tabbing}
\end{minipage}
\end{defn}

\noindent In this case, condition (2) ensures that an obligation prior to $q$
in the chain is not in force or has already been fulfilled (thus, no
reparation is required), while condition (3) states that there exists at least
one explicit derived permission prior to $q$.


\smallskip

\noindent We now introduce the proof conditions for $\pm\partial_{\OBL}$ and $\pm\partial_{\PERM}$.

\begin{defn}\label{def:+pO}

The proof condition of \emph{defeasible provability for obligation} is

\smallskip \noindent
\begin{minipage}{.3\textwidth}
\begin{tabbing}
  $+\partial_{\OBL}$: If $P(n+1)=+\partial_{\OBL} q$ then\\
  (1) \= $\OBL q \in \FACTS$ or\\
  \> (2.1) $\OBL \non q \not\in \FACTS$ and $\neg \OBL q \not\in \FACTS$ and $\PERM \non q \not\in \FACTS$ and\\ 
  \> (2.2) $\exists r\in R^{\OBL}[q,i]$ such that $r$ is applicable for $q$, and \\
  \> (2.3) \= $\forall s\in R[\non q, j]$, either \\
      \> \> (2.3.1) $s$ is discarded, or either\\
      \> \> (2.3.2) $s \in R^{\OBL}$ and $\exists t\in R[q,k]$ such that $t$ is applicable for $q$ and $t> s$, or\\
	  \> \> (2.3.3) $s \in R^{\PERM} \cup R_{def}$ and $\exists t\in R^{\OBL}[q,k]$ such that $t$ is applicable for $q$ and $t> s$.
	
\comment{\red{, and
  (3) \= $\forall s\in R^{\PERM}[\non q, j]$, either \\
	      \> (2.1) $s$ is discarded for $\non q$, or \\
	      \> (2.2) $\exists t\in R^{\OBL}[q,k]$ such that $t$ is applicable for $q$ and $t> s$.}}
\end{tabbing}
\end{minipage}

\end{defn}

\smallskip
\noindent To show that $q$ is defeasibly provable as an obligation, there are
two ways: (1) the obligation of $q$ is a fact, or (2) $q$ must be derived by
the rules of the theory. In the second case, three conditions must hold: (2.1)
$q$ does not appear as not obligatory as a fact, and $\non q$ is neither
provable as an obligation nor as a permission using the set of modal facts at
hand; (2.2) there must be a rule introducing the obligation for $q$ which can
apply; (2.3) every rule $s$ for $\non q$ is either discarded or defeated by a
stronger rule for $q$. If $s$ is an obligation rule, then it can be
counterattacked by any type of rule; if $s$ is a defeater or a permission
rule, then only an obligation rule can counterattack it.

The strong negation of the above definition gives us the negative proof
condition for obligation. Notice that the \emph{strong negation} of a formula
is closely related to the function that simplifies a formula by moving all
negations to an inner most position in the resulting formula, and replaces the
positive tags with the respective negative tags, and the other way around
\cite{ecai2000-5,igpl09policy}.


\begin{defn}\label{def:-pO}

The proof condition of \emph{defeasible refutability for obligation} is

\smallskip \noindent
\begin{minipage}{.3\textwidth}
\begin{tabbing}
  $-\partial_{\OBL}$: If $P(n+1)=-\partial_{\OBL} q$ then\\
  (1) \= $\OBL q \not\in \FACTS$ and either\\
  \> (2.1) $\OBL \non q \in \FACTS$ or $\neg \OBL q \in \FACTS$ or $\PERM \non q \in \FACTS$ or\\
  \> (2.2) $\forall r \in R^{\OBL}[q,i]$ either $r$ is discarded for $q$, or \\
  \> (2.3) \= $\exists s \in R[\non q, j]$ such that\\
      \> \> (2.3.1) $s$ is applicable for $\non q$, and \\
      \> \> (2.3.2) if $s \in R^{\OBL}$ then $\forall t \in R[q,k]$, either $t$ is discarded or $t\not> s$, and \\
	  \> \> (2.3.3) if $s \in R^{\PERM} \cup R_{def}$ then $\forall t \in R^{\OBL}[q,k]$, either $t$ is discarded or $t\not> s$.\\

\comment{\red{, and
  (3) \= $\forall s\in R^{\PERM}[\non q, j]$, either \\
	      \> (2.1) $s$ is discarded for $\non q$, or \\
	      \> (2.2) $\exists t\in R^{\OBL}[q,k]$ such that $t$ is applicable for $q$ and $t> s$.}}
\end{tabbing}
\end{minipage}

\end{defn}

\noindent We now introduce and briefly explain the proof conditions for permission.

\begin{defn}\label{def:+pP}

The proof condition of \emph{defeasible provability for permission} is

\smallskip \noindent
\begin{minipage}{.3\textwidth}
\begin{tabbing}
  $+\partial_{\PERM}$: If $P(n+1)=+\partial_{\PERM} q$ then\\
  (1) \= $\PERM q \in \FACTS$ or\\
  \> (2.1) $\OBL \non q \not\in \FACTS$ and $\neg \PERM q \not\in \FACTS$ and\\
  \> (2.2) $\exists r\in R^{\PERM}[q,i]$ such that $r$ is applicable for $q$, and \\
  \> (2.3) \= $\forall s\in R^{\OBL}[\non q, j]$, either \\
      \> \> (2.3.1) $s$ is discarded for $\non q$, or \\
      \> \> (2.3.2) $\exists t\in R[q,k]$ such that $t$ is applicable for $q$ and $t> s$.
\end{tabbing}
\end{minipage}

\end{defn}

\smallskip

\noindent This proof condition differs from its counterpart for obligation in
two aspects: we allow scenarios where both $+\partial_{\PERM} q$ and
$+\partial_{\PERM} \non q$ hold, but $+\partial_{\OBL} \non q$ must not hold
(clause 2.1); any applicable rule $s$ supporting $\non q$ can be counterattacked by
any type of rule $t$ supporting $q$, as $s$ must be an obligation rule, and
permission rules can only be attacked by obligation rules (clause 2.3).

As argued above, we define the negative proof condition for permission as the
strong negation of that for $+\partial_{\PERM}$.

\begin{defn}\label{def:-pP}

The proof condition of \emph{defeasible refutability for permission} is

\smallskip \noindent
\begin{minipage}{.3\textwidth}
\begin{tabbing}
  $-\partial_{\PERM}$: If $P(n+1)=-\partial_{\PERM} q$ then\\
  (1) \= $\PERM q \not\in \FACTS$ and either\\
  \> (2.1) $\OBL \non q \in \FACTS$ or $\neg \PERM q \in \FACTS$, or\\
  \> (2.2) $\forall r\in R^{\PERM}[q,i]$, either $r$ is discarded, or \\
  \> (2.3) \= $\exists s\in R^{\OBL}[\non q, j]$ such that \\
      \> \> (2.3.1) $s$ is applicable for $\non q$, and \\
      \> \> (2.3.2) $\forall t\in R[q,k]$, either $t$ is discarded or $t\not> s$.
\end{tabbing}
\end{minipage}

\end{defn}

\smallskip

\noindent The logic resulting from the above proof conditions enjoys
properties describing the appropriate behaviour of the modal operators. 

\begin{defn}\label{def:consistency}
A Defeasible Theory $D = (F, R, >)$ is \emph{consistent} iff $>$ is acyclic
and $\FACTS$ does not contain pairs of complementary (modal) literals, that is
if $D$ does not contain pairs like $\OBL l$ and $\neg\OBL l$, $\PERM l$ and
$\neg\PERM l$, or $l$ and $\non l$. The theory $D$ is
$O$-\emph{consistent} iff $>$ is acyclic and for any literal $l$, \FACTS does
not contain any of the following pairs: $\OBL l$ and $\OBL\non l$, $\OBL l$
and $\PERM\non l$.
\end{defn}

As usual, given a Defeasible Theory $D$, we will use $D\vdash\pm
\partial_{\Box}l$ iff there is a proof $P$ in $D$ such that $P(n) = \pm \partial_{\Box}l$ for an index $n$.
\begin{prop}
  \label{prop:consistence}
  Let $D$ be a consistent Defeasible Theory, and $\Box\in\MOD$.
  For any literal $l$, it is not possible to have both
  $D\vdash+\partial_{\Box}l$ and $D\vdash-\partial_{\Box}l$.
\end{prop}

\begin{proof}
It straightforwardly follows from the principle of strong negation proposed in
\cite{ecai2000-5,igpl09policy}: indeed, the negative proof tags proposed in
this work are defined as the strong negation of the positive ones.
\end{proof}

\noindent The meaning of the above proposition is that it is not possible to
prove that a literal is at the same time obligatory and not obligatory, or
permitted and not permitted.

\begin{prop} \label{prop:oconsistence}
  Let $D$ be an $O$-consistent Defeasible Theory. For any literal $l$, it
  is not possible to have both $D\vdash+\partial_{\OBL}l$ and
  $D\vdash+\partial_{\OBL}\non l$.
\end{prop}

\begin{proof}
We split the proof in two cases: (i) at least one of $\OBL l$ and $\OBL\non l$
is in $F$, and (ii) none of them is in $F$.

For (i) the proposition immediately follows by the assumption of
$\OBL$-consistency. Suppose that $\OBL l\in F$. Then clause (1) of
$+\partial_{\OBL}$ holds for $l$. By $\OBL$-consistency $\OBL\non l\notin F$,
thus clause (1) of $+\partial_{\OBL}$ does not hold for $\non l$. Since $\OBL
l\in F$, clause (2.1) of $+\partial_{\OBL}$ is always falsified for $\non l$,
and the thesis is proved.

For (ii): First of all, it is easy to verify that no rule
can be at the same time applicable and discarded for the derivation of
$\pm\partial_{\OBL}l (\non l)$. Then, since both $+\partial_\OBL l$ and
$+\partial_\OBL\non l$ hold, we have that there are applicable obligation
rules for both $l$ and $\non l$. This means that clause (2.3.2) holds for both
$+\partial_\OBL l$ and $+\partial_\OBL\non l$. Therefore, for every applicable
rule for $l$ there is an applicable rule for $\non l$ stronger than the rule
for $l$, and symmetrically, for every applicable rule for $\non l$ there is an
applicable rule for $l$ stronger than the rule for $\non l$. Since the set of
rules in a theory is finite, the situation we have just described is possible
only if there is a cycle in the transitive closure of the superiority
relation. Therefore, we have a contradiction because the superiority
relation is assumed to be acyclic (the transitive closure of the relation does
not contain cycles).
\end{proof}

\noindent The meaning of the proposition is that no formula is both obligatory
and forbidden at the same time. However, the proposition does not hold for
permission. It is possible to have both the explicit permission of $l$ and the
explicit permission of $\non l$.

The relationships between permissions and obligations are governed by the
following proposition:
\begin{prop}
  \label{prop:oblperm}
  Let $D$ be an $O$-consistent Defeasible Theory. For any literal $l$: 
  \begin{enumerate}
	\item\label{oblobl} 
    if $D\vdash+\partial_{\OBL} l$, then $D\vdash-\partial_{\OBL} \non l$;
    \item\label{oblperm} 
      if $D\vdash+\partial_{\OBL} l$, then $D\vdash-\partial_{\PERM} \non l$;
    \item\label{permobl} 
      if $D\vdash+\partial_{\PERM} l$, then $D\vdash-\partial_{\OBL} \non l$.
  \end{enumerate}
\end{prop}
\begin{proof}

\ref{oblobl}. Let $D$ be an $\OBL$-consistent Defeasible Theory, and
$D\vdash+\partial_{\OBL}l$. Literal $\non l$ can be in only one of the
following mutually exclusive situations: (i) $D\vdash +\partial_{\OBL} \non
l$; (ii) $D\vdash -\partial_{\OBL} \non l$; (iii) $D\not\vdash
\pm\partial_{\OBL} \non l$. Proposition~\ref{prop:oconsistence} allows us to
exclude case (i) since $D\vdash +\partial_{\OBL} l$ by hypothesis. Situation
(iii) denotes situations where there are loops in the theory involving literal
$\non l$\footnote{For examples situations like $\OBL \non l \To_{\OBL} \non l$, where the proof conditions will generate a loop without introducing a proof.}, but inevitably this would affect also the provability of literal
$l$, i.e., we would not be able to give a proof for $+\partial_{\OBL} l$ as
well. This is in contradiction with the hypothesis; thus, situation (ii) must
be the case.

\ref{oblperm}. Let $D$ be an $\OBL$-consistent Defeasible Theory, and
$D\vdash+\partial_{\OBL}l$. By definition of proof conditions for
$+\partial_{\OBL}$, statements (1)--(2.3.3) hold.

If $\OBL l \in \FACTS$, then condition (2.1) of $-\partial_{\PERM}$ holds for
$\non l$. Furthermore, by $\OBL$-consistency of $D$, $\OBL l \in \FACTS$
implies condition (1) of $-\partial_{\PERM}$ for $\non l$, and we obtain the
thesis.

Otherwise, from condition (2.2) of $+\partial_{\OBL}$,
conditions (2.3) and (2.3.1) of $-\partial_{\PERM}$ follow. Again, we can
iterate the same reasoning for condition (2.3.1) of $+\partial_{\OBL}$
implying condition (2.2) of $-\partial_{\PERM}$. It remains to consider
conditions (2.3.2) and (2.3.3) of $+\partial_{\OBL}$. In the first case, the
attacking rule $s$ is an obligation rule, thus it is of no interest in this
proof since in condition (2.2) of $-\partial_{\PERM}$ we consider only
permission rules. Thus, for the latter case, we know there exists a rule $t$
for obligation that is stronger than $s$, and this tuple of rules $(t,s)$ in
condition (2.3.3) for $+\partial_{\OBL}$ is the equivalent but opposite tuple
of the rules used in condition (2.3.2) for $-\partial_{\PERM}$. But, to
analyse in an exhaustive way condition (2.3.2) of $-\partial_{\PERM}$, we have
to take into consideration the whole set of rules $S=\{s_{1},\dots,s_{n}\}$
for permission leading to $\non l$, against the set of rules
$T=\{t_{1},\dots,t_{m}\}$ for obligation leading to $l$.\footnote{Notice that
the rules in conditions (2.2) and (2.3.3) are different rules: they form a
team that can defeat teams of rules for the opposite.} For each subset $S'$ of
$S$, every rule $s' \in S'$ is either discarded, or there exists a rule $t'
\in T$ that is stronger (the existence of $t'$ is guaranteed by the fact that
$+\partial_{\OBL} l$ holds by hypothesis). We can now remove $S'$ from $S$ (as
it cannot be used for proving $+\partial_{\PERM} \non l$), reducing the number
of elements in $S$. The iteration of this procedure eventually empties the set
$S$ since the number of rules in $D$ is finite, and the superiority
relation is acyclic.

\ref{permobl}. Let $D$ be an $\OBL$-consistent Defeasible Theory, and
$D\vdash+\partial_{\PERM}l$. By definition of proof conditions for
$+\partial_{\PERM}$, statements (1) or (2.1)--(2.3.2) hold. The proof follows
step by step that of Part 2 of the proposition.
Moreover, steps for conditions (1)--(2.2) are the mere juxtaposition. It
remains to analyse the interrelationship between condition (2.3) of
$+\partial_{\PERM}$ and conditions (2.3.2)--(2.3.3) of $-\partial_{\OBL}$.
Since condition (2.3) of $+\partial_{\PERM}$ takes into account only rules for
obligation which are systematically defeated with an analogous process of rule
elimination, thus conditions (2.3.2) and (2.3.3) of $-\partial_{\OBL}$ are
satisfied.
\end{proof}

The combination of Part \ref{oblperm} and \ref{permobl} of Proposition \ref{prop:oblperm} describes the consistency between obligation
and permission. Part \ref{permobl} also gives the
relationships between strong and weak permission. As we discussed in Section
\ref{sec:introduction}, a weak permission is a permission obtained from the
failure to derive the opposite obligation. This means that we have the weak
permission of $p$ when we have $-\partial_{\OBL}\non p$, and Part \ref{permobl} guarantees that we have it when
$+\partial_{\PERM}p$ holds.

\smallskip

We conclude this section showing how the logic developed hitherto works with
the example introduced at the end of Section
\ref{sec:three_concepts_of_permission}.

\paragraph{Example.} Let us recall the scenario reported at the end of Section
\ref{sec:three_concepts_of_permission}, and formally explain the conclusions
of the theory using applicability of rules and proof tags as defined above.
The primary obligation to call the ambulance is obtained (i.e., we derive
$+\partial_{\OBL} CallAmbulance$), but the obligation is violated as $\neg
CallAmbulance \in \FACTS$, making $r_{1}$ applicable for $\mathit{Help}$; rule $r_3$ is
applicable for literal $\neg \mathit{Help}$ and could attack $r_1$, but $r_1 > r_3$,
thus we have also $+\partial_{\OBL} \mathit{Help}$. Also $CallFiremen$ is derived as an
obligation but it is violated, thus rule $r_2$ is applicable for literal
$Extinguish$ in the condition for $+\partial_{\OBL}$. As $+\partial_{\OBL} Help$ holds, $r_3$
is applicable for $\neg Extinguish$, and $r_3 > r_2$, thus we derive
$+\partial_{\PERM} \neg Extinguish$.

\section{Algorithms for defeasible extension}\label{sec:algos}

We now present procedures and algorithms apt to compute the \emph{extension}
of a \emph{finite} Defeasible Theory, i.e., a theory where the set of facts
and rules is finite, in order to bound the complexity of the logic introduced
in the previous sections. The algorithms are based on the algorithm proposed
by Maher \cite{Maher2001} to show that Defeasible Logic has linear complexity;
the algorithms also incorporate the notion of inferiorly defeated rules
proposed by \cite{Lam.2011} to handle directly the superiority relation.

This section is divided in three main parts. In the first part we give the
formal definitions and introduce the notation adopted. The second part, which
contains the main body, describes the required computations: Algorithms
\ref{alg:defeasible} and \ref{alg:checkFacts} effectively compute the
defeasible extension of a Defeasible Theory given as an input, while
Algorithms \ref{alg:discard}, \ref{alg:modifyobl}, and \ref{alg:modifyperm}
are but auxiliary procedures that execute all the necessary operations due to
any modification of the extension. Each algorithm is followed by a technical
explanation on how it works. In the third part we present formal properties
that are meant to prove the computational results proposed in
Section~\ref{sec:CompRes}.

\subsection{Notation for the algorithms}

We introduce the notation relevant to our framework.

Given a Defeasible Theory $D$, $HB_{D}$ is the set of literals such that the
literal or its complement appears in $D$, where `appears' means that it is a
sub-formula of a modal literal occurring in the theory. The modal Herbrand
Base of $D$ is $HB=\set{\Box l |\ \Box \in\MOD, l\in HB_{D}}$. Accordingly, the extension of a Defeasible Theory is defined as follows.

\begin{defn}\label{def:extension}
Given a Defeasible Theory $D$, the \emph{defeasible extension} of $D$ is
defined as
\[E(D) = (+\partial_{\OBL}, +\partial_{\PERM}, -\partial_{\OBL},
-\partial_{\PERM}),
\]
where $\pm\partial_{\Box} = \set{l \in HB_{D}: D \vdash
\pm\partial_{\Box} l}$ with $\Box \in \MOD$. We define two 
Defeasible Theories $D$ and $D'$ to be \emph{equivalent} (in notation $D
\equiv D'$) if they have the same extensions, i.e., $E(D) = E(D')$.
\end{defn}

\noindent The next definition introduces two syntactical operations on the
consequent of rules used by the algorithms, whose meaning will be clear in the
remainder.

\begin{defn}\label{def:trunctaion-removal}
  Let $c_1=\opseq[\oslash]{a}{1}{l-1}$ and $c_2=\opseq[\oslash]{a}{i+1}{n}$
  be two (possibly empty) $\oslash$-expressions such that $a_{i}$ does not occur
  in them, and $c=c_{1}\oslash a_{i}\oslash c_{2}$ is an $\oslash$-expression.
  Let $r$ be a rule with form $A(r)\arrow c$. We define
  the operation of \emph{truncation} of the consequent $c$ at $a_{i}$ as:
  \[
    A(r)\arrow c!a_{i} = A(r)\arrow c_{1}\oslash a_{i}.
  \]
  We define the \emph{removal} of $a_{i}$ from the consequent $c$, $A(r)\arrow c\ominus
  a_{i}$, as:
\begin{displaymath}
	A(r)\arrow c\ominus a_{i} =
	\begin{cases}
		A(r)\Rightarrow_{\OBL} c_{1}\otimes c_{2} & \text{ if $r$ is } A(r)\Rightarrow_{\OBL} c_{1}\otimes a_{i}\otimes c_{2} \\
		A(r)\Rightarrow_{\OBL} c_{1}\odot c_{2} & \text{ if $r$ is } A(r)\Rightarrow_{\OBL} c_{1}\otimes a_{i}\odot c_{2} \\
		A(r)\Rightarrow_{\Box \in\MOD} c_{1}\odot c_{2} & \text{ if $r$ is } A(r)\Rightarrow_{\Box \in\MOD} c_{1}\odot a_{i}\odot c_{2} \\
		A(r)\Rightarrow_{\PERM} c_{2} & \text{ if $r$ is } A(r)\Rightarrow_{\OBL} a_{i}\odot c_{2}
\end{cases}
\end{displaymath}
\end{defn}

\noindent The next definition extends the concept of complement presented in
Section \ref{sec:logics} for modal literals, and it is used to establish
the logical connection among proved and refuted literals in our framework.

\begin{defn}
We define the \emph{complement} of a given literal $l$, denoted by
$\tilde{l}$, as:
\begin{enumerate}
\item If $l\in\LIT$, then $\tilde{l}=\set{\non l}$;
\item If $l=\OBL m$, then $\tilde{l}=\set{\neg\OBL m, \OBL\non m, \PERM \non m}$;
\item If $l=\neg\OBL m$, then $\tilde{l}=\set{\OBL m}$;
\item If $l=\PERM m$, then $\tilde{l}=\set{\neg\PERM m, \OBL\non m}$;
\item If $l=\neg\PERM m$, then $\tilde{l}=\set{\PERM m}$.
\end{enumerate}
\end{defn}

\noindent Given $\Box\in \MOD$, the sets $\pm\partial_{\Box}$
denote the global sets of defeasible conclusions (i.e., the set of literals
for which condition $\pm\partial_{\Box}$ holds), while $\partial_{\Box}^{\pm}$
are the corresponding temporary sets. Notice that the complement of $\neg
\PERM m$ does not include $\OBL m$ (and vice versa) because the failure to
derive $\PERM m$ cannot depend on the derivation of $\OBL m$, but rather on
the fact that $\OBL \non m$ is the case.

\subsection{Algorithms} 
\label{sub:algorithms}

We begin this subsection by reporting and explaining the three auxiliary
procedures used in the two main algorithms for the computation of the
extension of a logic.

\smallskip

\noindent Algorithm \ref{alg:discard} \textsc{Discard} performs all operations
related to when $-\partial_{\Box}l$ holds for a given literal $l$.

\begin{algorithm}[h!]
\caption{Discard}\label{alg:discard}
\begin{algorithmic}[1]
\Procedure{Discard}{$l \in \LIT,\Box \in \set{\PERM, \OBL}$}
  \State $\partial^-_\Box \gets \partial^-_\Box\cup\set{l}$
  \State $R \gets \set{A(r) \setminus \set{\neg\Box l}\hookrightarrow C(r)|\ r\in R} \setminus \set{r\in R|\ \Box l\in A(r)}$
  \State $> \gets > \setminus \set{(r,s),(s,r)\in>|\ \Box l\in A(r)}$
  \State $HB \gets HB \setminus \set{\Box l}$
\EndProcedure
\end{algorithmic}
\end{algorithm}

First of all, literal $l$ is placed in the local set of refuted literals with
modality $\Box$ (line 2). Furthermore, condition $-\partial_{\Box} l$ makes
literal $\neg\Box l$ provable, therefore it can be safely removed from all rules
where it appears as an antecedent, being that its contribution to a rule being
applicable or refuted has already been established. Since $l$ cannot be
proved with modality $\Box$, every rule containing $\Box l$ in its body is
discarded by clauses $(1.1)$ and $(1.3)$ of Definition~\ref{defn:DISC+pO+pP},
and thus we can remove such rules without affecting the conclusions that can
be derived from the theory (line 3). In addition, we remove all pairs
involving the rules from the superiority relation (line 4), and $\Box l$ from
the modal Herbrand Base (line 5).

\medskip

\noindent Algorithms \ref{alg:modifyobl} \textsc{ModifyObl} and
\ref{alg:modifyperm} \textsc{ModifyPerm} behave in a very similar way: both of
them modify the theory to accommodate the positive derivation of a modal
literal. They only differ on the kind of rules they manipulate (obligation and
permission rules, respectively).

\begin{algorithm}[h!]
\caption{ModifyObl}\label{alg:modifyobl}
\begin{algorithmic}[1]
\Procedure{ModifyObl}{$l \in \LIT$}
  	\State $\partial^+_{\OBL}\gets \partial^+_{\OBL} \cup \set{l}$
	\State $\partial^-_{\OBL}\gets \partial^-_{\OBL} \cup \set{\non l}$
	\State $\partial^-_{\PERM}\gets \partial^-_{\PERM} \cup \set{\non l}$
	\State $HB \gets HB \setminus \set{\OBL l, \OBL \non l, \PERM \non l}$
	\If{$\OBL l \not \in \FACTS$}
		\State $R \gets \set{A(r) \setminus \set{\OBL l, \neg \OBL \non l}\hookrightarrow C(r) |\ r\in R} \setminus \set{r\in R |\ A(r) \cap \widetilde{\OBL l} \not = \emptyset}$
		\State $> \gets > \setminus \set{(r,s), (s,r) \in > |\  A(r) \cap \widetilde{\OBL l} \not = \emptyset}$
	\EndIf
	\State $R\gets \set{ A(r) \hookrightarrow C(r) \ominus l |\
        r\in R^{\OBL}[l,n] \text{ for an index $n$}}$
	\State $R\gets \set{ A(r) \hookrightarrow C(r) \ominus \non l |\
	          r\in R^{\PERM}[\non l,n] \text{ for an index $n$}}$
	\State $R\gets \set{ A(r) \hookrightarrow C(r)! \non l \ominus \non l |\
		          r\in R^{\OBL}[\non l,n] \text{ for an index $n$}}$	
\EndProcedure
\end{algorithmic}
\end{algorithm}

\begin{algorithm}[h!]
\caption{ModifyPerm}\label{alg:modifyperm}
\begin{algorithmic}[1]
\Procedure{ModifyPerm}{$l \in \LIT$}
  	\State $\partial^+_{\PERM}\gets \partial^+_{\PERM} \cup \set{l}$
	\State $\partial^-_{\OBL}\gets \partial^-_{\OBL} \cup \set{\non l}$
	\State $HB \gets HB \setminus \set{\PERM l, \OBL \non l}$
	\If{$\PERM l \not \in \FACTS$}
		\State $R \gets \set{A(r) \setminus \set{\PERM l, \neg \OBL \non l}\hookrightarrow C(r) |\ r\in R} \setminus \set{r\in R |\ A(r) \cap \widetilde{\PERM l} \not = \emptyset}$
		\State $> \gets > \setminus \set{(r,s), (s,r) \in > |\  A(r) \cap \widetilde{\PERM l} \not = \emptyset}$
	\EndIf
	\State $R\gets \set{ A(r) \hookrightarrow C(r)! \non l \ominus \non l |\ 
		          r\in R^{\OBL}[\non l,n] \text{ for an index $n$}}$
	\State $R\gets \set{ A(r) \hookrightarrow C(r)! l |\
					          r\in R^{\PERM}[l,n] \text{ for an index $n$}}$	
\EndProcedure
\end{algorithmic}
\end{algorithm}

The input of both procedures is a literal $l$. As such, we add it to the
corresponding set of derived literals (line 2). Since $D\vdash
+\partial_{\OBL} l$ implies $D\vdash -\partial_{\OBL} \non l,
-\partial_{\PERM} \non l$ by Proposition~\ref{prop:oblperm} Part \ref{oblobl}
and \ref{oblperm}, and $D\vdash +\partial_{\PERM} l$ implies $D\vdash
-\partial_{\OBL} \non l$ by Proposition~\ref{prop:oblperm} Part \ref{permobl},
we also remove $\non l$ from the appropriate sets of refuted literals; then
the modal literal along with the set of its complementaries\footnote{Notice
that we do not remove any negative modal literal from HB by definition of
modal Herbrand Base.} are removed from the Modal Herbrand Base (lines 3--5 and
3--4, respectively). Lines 6--9 and 5--8, respectively, follow the same
reasoning of line 3 in Algorithm~\ref{alg:discard} \textsc{Discard}. Finally,
the rules of the theory are modified on account
of the modality the literal is derived with as well as the conditions
for the applicability of a rule given in Definitions
\ref{defn:APPL+pO}--\ref{defn:DISC+pO+pP} (lines 10--12 and 9--10,
respectively).

\begin{algorithm}[h!]
\caption{CheckFacts}\label{alg:checkFacts}
\begin{algorithmic}[1]
\Procedure{CheckFacts}{}
	\For{$l \in F$}
		\State $R \gets \set{A(r) \setminus \set{l}\hookrightarrow C(r) |\ r\in R} \setminus \set{r\in R |\ A(r) \cap \tilde{l} \not = \emptyset}$
		\State $> \gets > \setminus \set{(r,s), (s,r) \in > |\ A(r) \cap \tilde{l} \not = \emptyset}$
		\If{$l \in \LIT$}
			\State $R\gets \set{ A(r) \hookrightarrow C(r) ! l |\
		          r\in R^{\OBL}[l,n] \text{ for an index $n$}}$
		\EndIf
		\If{$l = \OBL m$}
			\State \Call{ModifyObl}{$m$}
			\EndIf
			\If{$l = \neg \OBL m$}
			\State $-\partial_{\OBL}\gets -\partial_{\OBL} \cup \set{m}$
			\State $HB \gets HB \setminus \set{\OBL m}$
			\State $R\gets \set{ A(r) \hookrightarrow C(r)! m \ominus m |\ 
			          r\in R^{\OBL}[m,n] \text{ for an index $n$}}$
		\EndIf
		\If{$l = \PERM m$}
			\State \Call{ModifyPerm}{$m$}
		\EndIf
		\If{$l = \neg \PERM m$}
			\State $-\partial_{\PERM}\gets -\partial_{\PERM} \cup \set{m}$
			\State $HB \gets HB \setminus \set{\PERM m, \OBL m}$
			\State $R\gets \set{ A(r) \setminus \hookrightarrow C(r) \ominus m |\ r\in R^{\PERM}[m,n] \text{ for an index $n$}}$		
		\EndIf
	\EndFor
\EndProcedure
\end{algorithmic}
\end{algorithm}

\medskip

Before describing how Algorithm~\ref{alg:checkFacts} works, let us
recall some concepts about the provability of a literal. Given a 
Defeasible Theory, a modal literal $\Box l \in \FACTS$ is trivially proved
with the corresponding modality by definition. Furthermore, we also stated
that a non-modal literal is proved within the theory if it is a fact.

Based on these facts, the procedure described in
Algorithm~\ref{alg:checkFacts} \textsc{CheckFacts} begins by removing all
factual literals from every rule where they appear as an antecedent; it also
removes all rules whose body contains a complementary literal (line 3). The
superiority relation is then modified in view of this operation (line 4).

From this point on, different operations are performed on account of which
kind of factual literal is considered.

\begin{enumerate}

\item If $l$ is a non-modal literal, we truncate the head of
all rules at $l$, where $l$ appears as an obligation  (lines 5--7);

\item if $l$ is a positive modal literal for obligation (lines 8--10) or
permission (lines 16--18), then Algorithm~\ref{alg:modifyobl}
\textsc{ModifyObl} (respectively Algorithm~\ref{alg:modifyperm}
\textsc{ModifyPerm}) is called to properly modify the theory. Notice that
operations in lines 7--8 of Algorithm~\ref{alg:modifyobl} \textsc{ModifyObl}
and 6--7 of Algorithm~\ref{alg:modifyperm} \textsc{ModifyPerm} are not
performed in this case, since they are equivalent to lines 3--4 of
Algorithm~\ref{alg:checkFacts} \textsc{CheckFacts};

\item if $l$ is a negative modal literal for obligation $\neg \OBL m$ (lines
11--15), then $-\partial_{\OBL} m$ holds, and clause (2) of
Definition~\ref{defn:DISC+pO+pP} makes all rules containing $m$ as an
obligation in their heads discarded for all literals after $m$. Hence, we
truncate all these chains at $m$, and then remove $m$ (line 14);

\item if $l$ is a negative modal literal for permission $\neg \PERM m$
(lines 19--23), then $-\partial_{\PERM} m$ holds in the theory. Thus, we
remove $m$ in every chain where $m$ appears as a permission (line 22).

\end{enumerate}

\smallskip

\noindent We conclude this section by presenting and describing
Algorithm~\ref{alg:defeasible} \mbox{\textsc{ComputeDefeasible}}, which represents
the main core for the computation of the defeasible extension of a theory.

\begin{algorithm}[h!]
\caption{ComputeDefeasible}\label{alg:defeasible}
\begin{algorithmic}[1]

\Require A defeasible theory $D$.
\Ensure	 The extension $E(D)$ of $D$.
	
\For{$\Box \in \set{\OBL, \PERM}$}
	\State $+\partial_{\Box} \gets \emptyset$
	\State $-\partial_{\Box} \gets \emptyset$
\EndFor
\State $R[l]_{infd} \gets \emptyset$ for each $l \in \LIT$
\State \Call{checkFacts}{}
\Repeat
  \State $\partial^{+}_{\Box} \gets \emptyset$
  \State $\partial^{-}_{\Box} \gets \emptyset$
  \For{$\Box l\in HB, \Box \in \set{\OBL, \PERM}$}
    \If{$R^{\Box}[l]=\emptyset$}
      \State \Call{Discard}{$l,\Box$}
    \EndIf
	\If{there exists $r \in R^{\OBL}[l,1]$ such that $A(r)= \emptyset$}
		\State $R[\non l]_{infd} \gets R[\non l]_{infd} \cup \set{s \in R[\non l] |\  r>s}$
      \If{$\set{s\in R[\non l] |\ s>r}=\emptyset$}
        \State \Call{Discard}{$\non l,\Box$}
        \If{$R[\non l] \setminus R[\non l]_{infd} = \emptyset$ and $\neg \OBL l \not\in \FACTS$}
          	\State \Call{ModifyObl}{$l$}
    		\EndIf
      \EndIf
	\EndIf
	\If{there exists $r \in R^{\PERM}[l,1]$ such that $A(r)= \emptyset$}
		\State $R[\non l]_{infd} \gets R[\non l]_{infd} \cup \set{s \in R^{\OBL}[\non l] |\  r>s}$
    \If{$\set{s\in R[\non l] |\ s>r}=\emptyset$}
      \State \Call{Discard}{$\non l, \OBL$}
    	\If{$R^{\OBL}[\non l] \setminus R[\non l]_{infd} = \emptyset$}
      		\State \Call{ModifyPerm}{$l$}
		  \EndIf
    \EndIf
	\EndIf
\EndFor
\State $+\partial_{\Box} \gets +\partial_{\Box} \cup \partial^{+}_{\Box}$
\State $-\partial_{\Box} \gets -\partial_{\Box} \cup \partial^{-}_{\Box}$
\Until{$\partial^{+}_{\Box}=\emptyset$ and $\partial^{-}_{\Box}=\emptyset$}
\State \Return $E(D) = (+\partial_{\OBL}, +\partial_{\PERM}, -\partial_{\OBL}, -\partial_{\PERM})$
\end{algorithmic}
\end{algorithm}

\noindent In lines 1--5, we initialise variables $+\partial_{\OBL}$,
$+\partial_{\PERM}$ and, for each literal $l$, a set $R[l]_{infd}$ containing
all the rules for $l$ that are defeated by a rule for the opposite.
Algorithm~\ref{alg:checkFacts} \mbox{\textsc{CheckFacts}} is invoked to compute all
defeasible conclusions derived from the set of facts (line 6).

The algorithm consists of a main loop (the \textbf{repeat} cycle in lines
7--35) that performs a series of transformations to reduce a Defeasible
Theory into a simpler equivalent one. The loop ends when no more modifications
on the extension are made, i.e., when both variables $\partial^{+}_{\Box}$ and
$\partial^{-}_{\Box}$ are empty at the end of an iteration.

At the beginning of the cycle, we re-initialise the set of conclusions computed at the iteration of the main loop (lines 8--9).

The \textbf{for} cycle in lines 10--32 checks all the rules for every literal
$l$ of the theory. In lines 11--13 it modifies the theory invoking
Algorithm~\ref{alg:discard} \textsc{Discard} for all modal literals with
no supporting chains. Lines 14--31 loop over all rules in the theory for the
current literal $l$, and checks if an applicable rule exists with $l$ as first
element in its head. If the rule introduces $l$ as an obligation (lines
14--22), then we have to collect all rules for the opposite, and check if they
are all defeated by a rule for $l$. If this is the case, then we have proved
$+\partial_{\OBL} l$ and Algorithm~\ref{alg:modifyobl} \textsc{ModifyObl} must
be invoked.

On the other hand, if $l$ is introduced as a permission (lines 23--31), then
we have to take into account only obligation rules for $\non l$, and to check
if every rule for $\non l$ as an obligation is defeated by at least one rule
for the opposite. If so, condition $+\partial_{\PERM} l$ holds, and
Algorithm~\ref{alg:modifyperm} \textsc{ModifyPerm} is invoked. Finally, all
modifications on the extension, due to the execution of the cycle, are stored
in the global sets of conclusions (lines 33--34).

\subsection{Properties of defeasible theory transformations}\label{subsec:Trans}

The properties we are going to show below are related to operations that
transform a theory $D$ into an equivalent simpler theory $D'$ (where by the
term `simpler' we mean a theory with a minor number of symbols in it).

The transformations operate either by removing some
elements from it, or by deleting a rule from the theory. Given the functional
nature of the operations, we will refer to the rules in the target theory with
the same names/labels as the rules in the source theory. Thus, given a
rule $r\in D$, we will refer to the rule corresponding to it in $D'$ (if it
exists) with the same label, namely $r$.

For the sake of readability, the proofs of all the theoretical results
(Lemmas~\ref{lem:facts}--\ref{lem:-p}) are not reported in this
subsection and the interested reader can find them in
Appendix~\ref{sec:Appendix}.

\medskip

\noindent Given a non-modal literal $p\in \FACTS$, we can obtain an equivalent
theory by removing $p$ in every rule where it appears in the antecedent.
Moreover, if the rule is for an obligation, Definition~\ref{defn:DISC+pO+pP}
clause (2) ensures that the rule will be discarded for every element after
$p$, and therefore we can truncate the reparation chain at $p$. Instead, if
the rule is for a permission, we cannot operate on it. In both cases, we only
consider rules where the complement of $p$ does not appear in the antecedent.
Finally, the superiority relation can be simplified by removing all tuples
with a rule containing $\non p$ in the antecedent, or an obligation rule for
an element after $p$ in its consequent.

\begin{lem}\label{lem:facts}
Let $D=(\FACTS,R,>)$ be a theory such that $p\in \FACTS \cap \LIT$. Let
$D'=(\FACTS ',R',>')$ be the theory obtained from $D$ where
\begin{align*}
    \FACTS ' = &\ \FACTS \setminus\set{p}\\
    R' = & \set{r: A(r)\setminus\set{p} \Rightarrow_{\OBL} C(r)!p |\
                r\in R,\ A(r)\cap\tilde{p}=\emptyset}\ \cup \\
		 & \set{r: A(r)\setminus\set{p} \Rightarrow_{\PERM} C(r) |\
                r\in R,\ A(r)\cap\tilde{p}=\emptyset} \\
    >' = & > \setminus \set{(r,s),(s,r)|\ r,s\in R, A(r)\cap\tilde{p}\neq\emptyset}.
  \end{align*}
Then $D\equiv D'$. 
\end{lem}

\noindent Starting from the modified theory given by the transformations of
the previous lemma, we now consider a theory with only modal literals in the
set of facts. If a literal $p$ is provable as an obligation, then we can
simplify the theory by removing $\OBL p$ in every antecedent of the rules in
$R$, and erase the rules where at least one element of $\widetilde{\OBL p}$ appears
in the antecedent. Since by hypothesis $\FACTS\cap\LIT=\emptyset$, if $p$
is present in the reparation chain of an obligation rule, we simplify the theory
by removing $p$ from the consequent. If $\non p$ is in the consequent, we can
also truncate the reparation chain of the rule since, by
Definition~\ref{defn:DISC+pO+pP} clause (2), the rule will be discarded for
each element after $\non p$ (Proposition~\ref{prop:oblperm} Part \ref{oblobl}
states that $-\partial_{\OBL}\non p$ holds as well). Moreover,
Proposition~\ref{prop:oblperm} Part \ref{oblperm} ensures that
$-\partial_{\PERM} \non p$ holds. Thus, Definition~\ref{defn:APPL+pP} clause
(3) allows us to remove $\non p$ in the consequent of permission rules for
$\non p$. Finally, the superiority relation can be simplified by removing all
tuples with a rule containing at least one element of $\widetilde{\OBL p}$ in
the antecedent.

\begin{lem}\label{lem:obl-alg}
  Let $D=(\FACTS,R,>)$ be a theory such that  $\FACTS\cap\LIT=\emptyset$ and 
  $D\vdash+\partial_{\OBL}p$. Let $D'=(F,R',>')$ be the theory obtained from $D$
  where
  \begin{align*}   
    R' = & \set{r: A(r)\setminus\set{\OBL p} \Rightarrow_{\OBL} C(r) ! \non p \ominus \non p |\ r\in R,\ A(r)\cap\widetilde{\OBL p}=\emptyset}\ \cup\\
& \set{r: A(r)\setminus\set{\OBL p} \Rightarrow_{\OBL} C(r) \ominus p |\ r\in R,\ A(r)\cap\widetilde{\OBL p}=\emptyset}\ \cup \\
& \set{r: A(r)\setminus\set{\OBL p} \Rightarrow_{\PERM} C(r) \ominus \non p |\ r\in R,\ A(r)\cap\widetilde{\OBL p}=\emptyset}\\
    >' & = > \setminus \set{(r,s),(s,r)|\ r,s\in R, 
      A(r) \cap \widetilde{\OBL p}\neq\emptyset}.
  \end{align*} 
  Then $D\equiv D'$.  
\end{lem}

\noindent As the previous lemma, we consider a theory with only modal
literals in the set of facts. Since the theory proves
$-\partial_{\OBL} p$, also $\neg \OBL p$ holds. Thus, we
obtain an equivalent simpler theory by erasing all rules with $\OBL p$ as one
of the antecedents, and by removing $\neg \OBL p$ in each rule where it
appears in the antecedent. Again, by Definition~\ref{defn:DISC+pO+pP} clause
(2), for every obligation rule we can truncate each reparation chain with $p$
in the consequent and eliminate it from such a chain. Finally, the superiority
relation can be simplified by removing all the pairs with a rule containing
$\OBL p$ in the antecedent.

\begin{lem}\label{lem:notObl-alg}
  Let $D=(F,R,>)$ be a theory such that  $F\cap\LIT=\emptyset$ and 
  $D \vdash -\partial_{\OBL} p$. Let $D'=(F,R',>')$ be theory obtained from $D$
  where
  \begin{align*}   
    R'  =  &\set{r: A(r)\setminus\set{\neg \OBL p} \Rightarrow_{\Box } C(r) |\ r\in R,\ A(r)\cap\set{\OBL p}=\emptyset}\ \cup \\  
& \set{r: A(r) \Rightarrow_{\OBL} C(r) ! p \ominus p |\ r\in R,\ A(r)\cap\set{\OBL p}=\emptyset} \\
    >'  = & > \setminus \set{(r,s),(s,r)|\ r,s\in R, 
      A(r) \cap \set{\OBL p}\neq\emptyset}.
  \end{align*} 
  Then $D\equiv D'$.    
\end{lem}

\noindent We can defeasibly prove a literal $p$ as a permission. A simpler
equivalent theory is one where we remove $\PERM p$ in each set of antecedents
and where we erase all the rules containing at least one element of the
complement of $\widetilde{\PERM p}$ in the antecedent. Proposition~\ref{prop:oblperm}
Part \ref{permobl} states that $-\partial_{\OBL} \non p$ holds. Thus, by
Definition~\ref{defn:DISC+pO+pP} clause (2), if $\non p$ appears in the
reparation chain of an obligation rule, we can remove it after having truncate
the chain at $\non p$. Instead, if we consider permission rules with $p$ in
the consequent, by Definition~\ref{defn:DISC+pO+pP} clause (3), we can
truncate the corresponding reparation chain at $p$. Finally, the superiority
relation can be simplified by removing all the pairs with a rule with an
element of $\widetilde{\PERM p}$ in the antecedent.

\begin{lem}\label{lem:perm-alg}
  Let $D=(F,R,>)$ be a theory such that  $F\cap\LIT=\emptyset$ and 
  $D\vdash+\partial_{\PERM}p$. Let $D'=(F,R',>')$ be theory obtained from $D$
  where
  \begin{align*}   
    R' =  &\set{r: A(r)\setminus\set{\PERM p} \Rightarrow_{\OBL} C(r)! \non p \ominus \non p |\
                r\in R,\ A(r)\cap\widetilde{\PERM p}=\emptyset}\ \cup \\
		 & \set{r: A(r)\setminus\set{\PERM p} \Rightarrow_{\PERM} C(r) ! p |\
                r\in R,\ A(r)\cap\widetilde{\PERM p}=\emptyset} \\
    >' = & > \setminus \set{(r,s),(s,r)|\ r,s\in R, 
      A(r)\cap\widetilde{\PERM p}\neq\emptyset}.
  \end{align*} 
  Then $D\equiv D'$.    
\end{lem}

\noindent The theory proves $-\partial_{\PERM}
p$, allowing $\neg \PERM p$ to hold. Thus, we obtain an equivalent theory if we
erase all the rules with $\PERM p$ in the set of the antecedents and if we
remove $\neg \PERM p$ where it appears in the tail of a rule. Moreover, if the
rule is for a permission, we remove $p$ from the reparation chain.
Finally, we change the superiority relation by erasing the tuples with a rule with $\PERM p$ in the antecedent.

\begin{lem}\label{lem:notPerm-alg}
  Let $D=(F,R,>)$ be a theory such that  $F\cap\LIT=\emptyset$ and 
  $D\vdash-\partial_{\PERM}p$. Let $D'=(F,R',>')$ be theory obtained from $D$
  where
  \begin{align*}   
    R' =  &\set{r: A(r)\setminus\set{\neg \PERM p} \Rightarrow_{\Box } C(r) |\ r\in R,\ A(r)\cap\set{\PERM p}=\emptyset}\ \cup \\ 
& \set{r: A(r) \Rightarrow_{\PERM} C(r) \ominus p |\ r\in R,\ A(r)\cap\set{\PERM p}=\emptyset}\\
    >'  = &> \setminus \set{(r,s),(s,r)|\ r,s\in R, 
      A(r)\cap\set{\PERM p}\neq\emptyset}.
  \end{align*} 
  Then $D\equiv D'$.    
\end{lem}

\noindent The following two lemmas represent conditions under which a literal
can be proved as an obligation or as a permission. The transformations
dictated by the previous lemmas empty the antecedent of every applicable rule.

\begin{definition}\label{def:infd}
  Given a theory $D=(F,R,>)$, and a set of rules $S$, the subset of $S$ of 
  \emph{inferiorly defeated} rules for a literal $p$, $S[p]_{infd}$ is thus 
  defined:
  \(
    r\in S[p]_{infd}
  \) iff
  \begin{enumerate}
    \item $\exists s\in R[\non p]$ such that $A(r)=\emptyset$ and $s>r$, and
    \item if $r\in R^{\PERM}[p]$, then $s\in R^{\OBL}[\non p]$.
  \end{enumerate}
\end{definition}

\begin{lem}\label{lem:proveOBL}
  Let $D=(F,R,>)$ be a theory such that $F\cap\MODLIT=\emptyset$, $\exists r\in 
  R^{\OBL}[p,1]$, $A(r)=\emptyset$, and $R[\non p]\subseteq R_{infd}$. 
  Then $D\vdash+\partial_{\OBL}p$.
\end{lem}

\begin{lem}\label{lem:provePERM}
  Let $D=(F,R,>)$ be a theory such that $F\cap\MODLIT=\emptyset$, $\exists r\in 
  R^{\PERM}[p,1]$, $A(r)=\emptyset$, and $R^{\OBL}[\non p]\subseteq R_{infd}$.
  Then $D\vdash+\partial_{\PERM}p$.
\end{lem}
\noindent The next two lemmas concern conditions to determine when it is possible to 
assert that a literal is negatively provable.
\begin{lemma}\label{lem:-pEmpty}
  Let $D=(F,R,>)$ be a theory such that $F\cap\MODLIT=\emptyset$ and 
  $R^{\Box }[p]=\emptyset$, for $\Box \in\set{\OBL,\PERM}$. Then $D\vdash-\partial_\Box p$.
\end{lemma}

\begin{lemma}\label{lem:-p}
  Let $D=(F,R,>)$ be a theory such that $F\cap\MODLIT=\emptyset$, and $\exists r\in 
  R[p,1]$ such that $A(r)=\emptyset$ and $r_{sup}=\emptyset$. Then
  \begin{enumerate}
    \item if $r\in R^{\OBL}$, then $D\vdash-\partial_{\Box}\non p$, 
      $\Box\in\set{\OBL,\PERM}$;
    \item if $r\in R^{\PERM}\cup R_{def}$, then $D\vdash-\partial_{\OBL}\non p$. 
  \end{enumerate}
\end{lemma}

\section{Computational results}\label{sec:CompRes}

In this section we present the computational properties of the algorithms
previously described. Since, as stated above, the first three algorithms are
sub-routines of the two main ones, we will present the correctness and
completeness results of these algorithms inside theorems for Algorithms
\ref{alg:checkFacts} \textsc{CheckFacts} and \ref{alg:defeasible}
\textsc{ComputeDefeasible}.

In order to properly exhibit results on the complexity of the algorithms, we
need the following definition.

\begin{defn}\label{defn:ComplexityOfATheory}
	Given a Defeasible Theory $D$, the \emph{size} $S$ of $D$ is the number of
literal occurrences plus the number of the rules in $D$.
\end{defn}

\noindent We also report some key ideas and intuitions behind our
implementation.

\begin{enumerate}

\item Each operation on global sets $\pm\partial_{\Box}$ and
$\partial^{\pm}_{\Box}$ requires a constant time, as we manipulate finite sets
of literals;

\item For each literal $\Box l \in HB$, we implement a hash table with
pointers to rules where the literal occurs; thus, retrieving the set
of rules containing a given literal requires constant time.

\item The superiority relation can also be implemented by means of hash
tables; once again, the information required to modify a given tuple can be
accessed in constant time.

\end{enumerate}

\begin{thm}\label{thm:ComplALGOCheckFacts}
	Given a modal Defeasible Theory $D$ with size $S$,
Algorithm~\ref{alg:checkFacts} \mbox{\textsc{CheckFacts}} terminates and its
computational complexity is $O(S)$.
\end{thm}

\begin{proof}

Termination of Algorithm~\ref{alg:checkFacts}~\textsc{CheckFacts} is given by
definition of modal Defeasible Theory, since the internal sub-routines (i.e.,
Algorithm~\ref{alg:modifyobl} \textsc{ModifyObl} and \ref{alg:modifyperm}
\textsc{ModifyPerm}), as well as the algorithm itself, manipulate finite sets
of rules and facts.

For a correct analysis of the complexity of Algorithm~\ref{alg:checkFacts}
\mbox{\textsc{CheckFacts}}, it is of the utmost importance to correctly
comprehend Definition~\ref{defn:ComplexityOfATheory}. Here we underline that
the size $S$ of a theory represents the total number of occurrences of each
literal in every rule of such a theory. Let us examine a theory with $Y$
literals and whose size is $Z$. If we consider a cycle whose purpose is to
call, for each literal, a procedure that selectively deletes it from all the
rules of the theory (no matter to what end), a rough computational complexity
would be $O(Z^{2})$. In fact, the complexity of the procedure by itself is
bound to the number of rules in the theory, which is in $O(Z)$, and this
procedure is iterated $Y \in O(Z)$ times. 

However, a more fined-grained analysis shows that the complexity
of this loop is lower. The mistake of the previous analysis is that it
considers the complexity of the procedure separately from the complexity of
the external loop, whilst they are strictly dependent. Indeed, the overall
number of operations made by the sum of all loop iterations cannot outrun the
number of occurrences of the literals, $O(Z)$, because the operations in the
inner procedure directly decrease, iteration after iteration, the number of
the remaining repetitions of the outmost loop, and the other way around.
Therefore, the overall complexity is not bound to $O(Z)\cdot O(Z) = O(Z^{2})$,
but to $O(Z) + O(Z) = O(Z)$.

We can contextualise the above reasoning to
Algorithm~\ref{alg:checkFacts}~\textsc{CheckFacts}. The main cycle in lines
2--26 is iterated over the set of facts, whose cardinality is in
$O(S)$; the operations in lines 10 and 19 (invoking
Algorithms~\ref{alg:modifyobl}~\textsc{ModifyObl} and
\ref{alg:modifyperm}~\textsc{ModifyPerm}) represent an additive factor $O(S)$
in the overall complexity of the algorithm. Finally, all operations on the set
of rules and the superiority relation performed by conditions in lines 5, 12,
and 21 require constant time, given the implementation of data structures
proposed above. Therefore, we can state that the complexity of the algorithm
is $O(S)$.
\end{proof}

\begin{thm}\label{thm:ComplALGOComputeDefeasible}
	Given a Defeasible Theory $D$ with size $S$,
Algorithm~\ref{alg:defeasible}\newline \mbox{\textsc{ComputeDefeasible}}
terminates and its computational complexity is $O(S)$.
\end{thm}

\begin{proof}
When referring to the termination of
Algorithm~\ref{alg:defeasible}~\textsc{ComputeDefeasible}, the most important
part we have to analyse is the \textbf{repeat} cycle in lines 7--35. Once an
instance of the cycle has been performed, we must be in one of the following
(mutually exclusive) situations:

\begin{enumerate}

\item no modification of the extension has occurred. In this case, line 29
ensures the termination of the algorithm;

\item the theory is modified with respect to a literal in the Modal Herbrand
Base $HB$. Notice that the algorithm takes care of removing the literal from
$HB$ once it has performed the suitable operations. As the set is finite, the
process described above eventually empties $HB$, and at the next iteration of
the cycle we have no means to modify the extension of the theory. In this case
as well, the algorithm ends its execution.	
\end{enumerate}

\noindent The analysis of the complexity of Algorithm~\ref{alg:defeasible}
\textsc{ComputeDefeasible} straightly follows from the reasoning proposed to
demonstrate the computational complexity of
Algorithm~\ref{alg:checkFacts}~\textsc{CheckFacts}. Thus, the \textbf{repeat}
cycle in lines 7--35 is in $O(S)$, while procedure invocations in lines 12,
17, 19, 26 and 28 represent an additive factor as before. Since the operations
in lines 15 and 24, and the checks in lines 14, 16, 18, 23, 25 and 27 also
weight a constant time, the computational complexity of
Algorithm~\ref{alg:defeasible} \textsc{ComputeDefeasible} is bound to $O(S)$.
\end{proof}

\begin{thm}\label{thm:SoundCompl}
	Algorithm~\ref{alg:defeasible}~\textsc{ComputeDefeasible} is sound and
complete.
\end{thm}

\begin{proof}
As already argued at the beginning of the section, the aim of
Algorithm~\ref{alg:defeasible} \mbox{\textsc{ComputeDefeasible}} is to compute
the defeasible extension of a given theory $D$ through successive
transformations on the set of facts and rules, and on the superiority
relation. These transformations act in ways which obtain at each step a new
simpler theory while retaining the same extension. Again, we remark that the
word `simpler' is used to denote a theory with less elements in it. Since we
have already proved the termination of the algorithm, it eventually comes to a
fix-point theory where no more operations can be made.

In order to demonstrate the soundness of Algorithm~\ref{alg:defeasible}, we
show in the list below that all the operations performed by the algorithm are
those described in Lemmas~\ref{lem:facts}--\ref{lem:-p}, where we have already
proved the soundness of the operation involved.

\setlist[itemize,2]{label=--}
\begin{itemize}
	\item Algorithm~\ref{alg:discard} \textsc{Discard}:
	\begin{itemize}
		\item Lines 2--4: Lemma~\ref{lem:notObl-alg} and \ref{lem:notPerm-alg}.
	\end{itemize}
	\item Algorithm~\ref{alg:modifyobl} \textsc{ModifyObl}:
	\begin{itemize}
		\item Lines 2--4, 10--13: Lemma~\ref{lem:obl-alg};
		\item Lines 6--9: Proposition~\ref{prop:oconsistence} and Lemma~\ref{lem:obl-alg}.
	\end{itemize}
	\item Algorithm~\ref{alg:modifyperm} \textsc{ModifyPerm}:
	\begin{itemize}
		\item Lines 2--3, 9--11: Lemma~\ref{lem:perm-alg};
		\item Lines 5--8: Proposition~\ref{prop:consistence} and Lemma~\ref{lem:perm-alg}.
	\end{itemize}
	\item Algorithm~\ref{alg:checkFacts} \textsc{CheckFacts}:
	\begin{itemize}
		\item Lines 3--4: Lemmas~\ref{lem:facts}--\ref{lem:notPerm-alg};
		\item Lines 5--8: Lemma~\ref{lem:facts};
    \item Lines 9--11: Algorithm~\ref{alg:modifyobl};
		\item Lines 12--17: Lemma~\ref{lem:notObl-alg};
		\item Lines 18--20: Algorithm~\ref{alg:modifyperm};
		\item Lines 21--25: Lemma~\ref{lem:notPerm-alg}.
	\end{itemize}
	\item Algorithm~\ref{alg:defeasible} \textsc{ComputeDefeasible}:
	\begin{itemize}
	  \item Lines 11--13: Lemma~\ref{lem:-pEmpty};
	  \item Lines 17 and 26: Lemma~\ref{lem:-p};
    \item Lines 18--20: Lemmas~\ref{lem:proveOBL} and \ref{lem:obl-alg};
    \item Lines 27--29: Lemma~\ref{lem:provePERM} and \ref{lem:perm-alg}.
	\end{itemize}
\end{itemize}

\noindent These results state that if in the initial theory a literal is
either defeasibly proved or not, so it will be in the final theory; thus
proving the soundness of the algorithm.

Moreover, since all lemmas show the equivalence of the two theories, and since
the equivalence relation is a bijection, this also gives the completeness of
Algorithm~\ref{alg:defeasible} \textsc{ComputeDefeasible}.
\end{proof}

\section{Discussion: The Three Types of Permission}
\label{sec:discussion}

Resuming the discussion of Section \ref{sec:three_concepts_of_permission}, we
will delve into the technical aspects of the three mentioned concepts of
permissions within our framework.

The idea of weak or negative permission is easily represented in the system as
follows:

\begin{definition}[Weak Permission]
\label{def:weak_permission}
Let $D$ be a Defeasible Theory. A literal $l$ is weakly permitted in $D$ iff 
$D\vdash -\partial_{\OBL} \non l$.
\end{definition}


One remark is in order here: Definition \ref{def:weak_permission} is useful to
check whether a literal $l$ is weakly permitted within the theory, but it
cannot be directly used to explicitly derive $\PERM l$ for triggering any rule
where this modal literal occurs in the antecedent. In fact, when
$\PERM l$ appears in the antecedent of a rule, then the only way to activate
such a rule is to explicitly derive $\PERM l$.

However, weak permissions are decisive in the applicability of a rule for
conditions (1.2) of Definitions~\ref{defn:APPL+pO} and~\ref{defn:APPL+pP};
when $\neg \OBL l$ occurs in the antecedent of the rule, then the theory must
prove $-\partial_{\OBL} l$. This is equivalent to say that $\non l$ is weakly
permitted in $D$.



This reading assumes that the distinction between weak and strong
permission goes beyond the idea, defended by \cite{alchourron-bulygin:1984},
that there is only one prescriptive type of permission. If the reader finds
our proposal limiting, we can trivially revise the rule applicability
conditions at point (1.3) (and adjust the algorithms accordingly) by
establishing that, when $\PERM l$ occurs in the antecedent of any rule $r$,
$r$ is applicable if one of the following conditions holds: (i)
$+\partial_{\PERM} l$, or (ii) $-\partial_{\OBL} \non l$ (observe that
$+\partial_{\PERM} l$ implies $-\partial_{\OBL} \non l$, but not vice versa).

A straightforward result (from Proposition \ref{prop:oblperm} Part
\ref{oblobl}) regarding weak permissions follows:

\begin{prop}\label{th:ought-can}
Let $D$ be any O-consistent Defeasible Theory. For any literal $l$, if \mbox{$D\vdash
+\partial_{\OBL} l$}, then $l$ is weakly permitted.
\end{prop}
As expected, weak permission enjoys the deontic principle ``Ought implies
Can'', i.e., the principle that in deontic logic is $\OBL l \to \PERM l$.

\medskip

\noindent We now consider the two ways to obtain strong permissions
in Defeasible Logic: using either explicit permissive norms or
defeaters.

The first case is naturally captured in the logical framework proposed in
Sections \ref{sec:logics} and \ref{sec:algos}. In the simplest case, a literal
like $\PERM l$ is derived in a theory $D$ when there is a successful
reasoning chain in which the last rule has the form $\seq{a} \To_{\PERM} l$.

More complex cases are due to the fact that $l$ may occur in an $\oslash$-expression. In this case $l$ is obtained as strongly
permitted if, for each literal $c$ preceding $l$ in the sequence, one of the
following conditions hold:
\begin{itemize}
\item if $c$ leads to derive $\OBL c$ (i.e., $c$ occurs in an
$\otimes$-subsequence of the main sequence where $l$ occurs), then this
obligation must be obtained and violated;
\item if $c$ leads to derive $\PERM c$ (i.e., $c$ occurs in an
$\odot$-subsequence of the main sequence where $l$ occurs), then this
permission is successfully attacked by an opposite obligation.
\end{itemize}
The introduction of sequences of permissions and obligations
enriches the language in a significant way, since it allows us to express a
preference among obligations and permissions when they are logically compatible. In the case of sequences of positive permissions, an
$\odot$-sequence states that a permissive exception of an obligation is
preferred with respect to another possible exception of the same obligation.
However, this extension does not conceptually change the fundamental intuition
that is also behind the basic case where permissive norms have the form
$\seq{a} \To_{\PERM} l$: the antecedent of positive permissive rules with head
$l$ provides sufficient defeasible reasons to derive $\PERM l$.

The second method considered in Section
\ref{sec:three_concepts_of_permission} to capture the notion of strong permissions acting as
exceptions to obligations looks at permissions as \emph{undercutters}
in argumentation theory (this idea was discussed in
\cite{boella-torre-nrac:2003}): a permissive norm with head $l$ operates in
such a way that, if applicable, it is not a sufficient reason for
deriving neither $l$, nor $\non l$, but it is a sufficient reason for
blocking the derivation of $\non l$ as obligatory. In Defeasible Logic, this idea is
naturally implemented by using defeaters. For the sake of simplicity, we have not considered this
concept of strong permission in Sections \ref{sec:logics} and \ref{sec:algos}.
However, to cover this case it is
sufficient to adopt one the following definitions (compare the definition of
$R^{\PERM}[q,n]$ in Section \ref{sec:logics}):
\begin{definition}\label{def:defeater-1}
The set $R^{\PERM}[q,n]$ is $X\cup Y$ where
\begin{itemize}
\item $X$ is the set of rules where $q$ appears at index $n$, and the operator $\odot$ precedes $q$;
\item $Y$ is the set of defeaters with head $q$.\footnote{In this case, $n$ is
always $1$.}
\end{itemize}
\end{definition}
\begin{definition}\label{def:defeater-2}
The set $R^{\PERM}[q,n]$ is the set of defeaters with head $q$.\footnote{Since
$n$ is always $1$, $n$ can be omitted and we can simply write $R^{\PERM}[q]$.}
\end{definition}
Definition \ref{def:defeater-1} adds the defeaters to the set of rules that
can be used to derive tagged literals like $+\partial_{\PERM} l$, obtaining
modal literals like $\PERM l$. Definition \ref{def:defeater-2} restricts
derivations of strong permissions only to reasoning chains where the last rule
is a defeater. In both cases, except these new definitions, we do not need to
change anything else in the logic (hence, in the proof conditions for
$\pm\partial_{\PERM}$) or the algorithms.

What is the difference between strong permissions obtained via
rules for permission and the ones obtained via defeaters? Although rules for
$\PERM$ and defeaters are not in general equivalent, as we have informally
suggested in Section \ref{sec:three_concepts_of_permission}, they behave quite
similarly when they are used to derive permissions, as well as to attack obligation rules supporting the opposite conclusion. In other words, defeaters
do not clash with any permissive rules. Consequently, if this reading of
defeaters is simply embedded within the proof conditions for
$\pm\partial_{\PERM}$ and for $\pm\partial_{\OBL}$ by adopting either
Definition \ref{def:defeater-1}, or Definition \ref{def:defeater-2}, then
rules for $\PERM$ and defeaters play a very similar role in the proof theory. In fact, if we consider condition (2.3.3) in the
proof condition for $\pm\partial_{\OBL}$, then two rules like
$r_1:~\seq{a}\To_{\PERM} l$ and $r_2: \seq{a}\defeater l$ both attack any rule
$s$ for obligation supporting $\non l$, and $s$ can counterattack
$r_1$ and $r_2$ as well.

%

We remark that the significant difference between the rules for $\PERM$ and
defeaters is that defeaters do not allow for having sequences of permissions
in their head. 


Finally, notice that neither type of strong permission considered enjoys the principle ``Ought implies Can''. This result comes directly from
Proposition \ref{prop:oblperm} and is based on the idea that the only manner
to derive strong permissions is by means of reasoning chains where the last
rule occurring in them is either a rule for $\PERM$ or a defeater (i.e.,
explicit permissive norms).


%
%
%
%
%
%

\section{Summary and Related Work} 
\label{sec:conclusions}

In this paper we proposed an extension of Defeasible Logic to represent
three concepts of defeasible permission. In particular, we examined
different types of explicit permissive norms that work as exceptions to
opposite obligations. We also discussed how strong permissions can be
represented with or without introducing a new consequence relation for
inferring conclusions from explicit permissive norms. Finally, we combined a
preference operator applicable to contrary-to-duty obligations with a new one
representing ordered sequences of strong permissions which derogate from
prohibitions. Special attention was devoted to the computational aspects of
the logic.

Although logicians have mostly overlooked the concept of permission over time,
the history of deontic logic offers some well-known key ideas to interpret it.
Indeed, the original intuition (proposed by \cite{wright:1951}, among others)
that permissions are the modal dual of obligations, though technically simple
and attractive, proved to be partial and simplistic (for a discussion, see
\cite{vonwright:1963,alchourron-bulygin:1984,alchourron:1993}). Hence,
subsequent contributions have helped to expand the picture in several directions.
The distinction between weak (or negative) and strong (or positive) permission
\cite{vonwright:1963} plays an important role in this regard. Though, Alchourr\'on and
Bulygin \cite{alchourron-bulygin:1981,alchourron-bulygin:1984} argued
that there is only one prescriptive sense of permission, while the distinction
between weak and strong permission makes sense only at a descriptive level,
depending on how any permission is obtained within a system of norms. Legal
theorists such as Alf Ross and Norberto Bobbio \cite{ross:1968,bobbio:1958}
claimed that legal permissions are in fact exceptions to obligations imposing
the opposite, even though this did not lead them (Ross, in particular) to
clearly link the concept of exception with the one of strong permission. Other
theorists even denied the usefulness of seeing strong permissions as
exceptions \cite{opalek-wolenski:1991,royakkers-dignum:1997}, since the former
ones introduce nothing but strong indifference in normative systems. This
thesis was instead rejected by \cite{alchourron-bulygin:1981}.

The purpose of this paper was not to provide a comprehensive logical theory of
permission, nor to take an exhaustive critical position in the debate that we
have very briskly recalled. Our goal was twofold:
\begin{itemize}
	\item to capture some aspects of permissions within a broader view of
defeasible normative reasoning;
	\item to study the defeasibility of permissions in a computationally
efficient logical framework.
\end{itemize}
At a more theoretical level, our work shares with
\cite{makinson-torre:2003,boella-torre-icail:2003,stolpe-jal:2010} some
conceptual assumptions. Slightly rephrasing \cite{stolpe-jal:2010}'s analysis,
the following guidelines inspired our treatment of the concept of permission:
\begin{enumerate}
	\item ``no logic of norms without attention to the system of which they form part'' \cite{makinson:1999}: our investigation of the concept of permission looks at how permissive norms and other types of norms interact within systems;
	\item the distinction between positive and negative permission is meaningful;
	\item one fundamental role of positive permissions is that of stating
exceptions to obligations; accordingly, positive permissions are supposed to
override, or at least block, some deontic conclusions coming from other norms;
	\item the logical space of weak permission is the one left unregulated by
mandatory norms.
\end{enumerate}
However, \cite{makinson-torre:2003,boella-torre-icail:2003,stolpe-jal:2010}
are all based on a different logical formalism, Input/Output Logic (IOL)
\cite{makinson00input}, thus it is difficult to compare in detail those
contributions with the present one. Normative reasoning is viewed in IOL as a
rule-based process of manipulation of inputs (factual premises) into outputs
(normative conclusions). The analysis of normative systems consists in
representing conditional norms simply as ordered pairs $(a,b)$ where $a$
represents the antecedent of the rule, and $b$ its consequent: ``if $a$ then
$b$'' where $a$ has factual content and $b$ normative content, viz. an
obligation or a permission. Typically, both $a$ and $b$ are taken to be
formulas from propositional logic. Each set of such ordered pairs can be seen
as an inferential mechanism which, given an input, determines an output based
on those connections. Formally, given a set of positive mandatory norms
(obligations) $G$ and a set of permissive norms (positive permissions) $P$, a
closure operation $C$, and a set of facts $A$, the output of $G\cup P$ given
$C$ and a set of input formulas is $\mathit{out}_C(G\cup P,A) = \{ b \mid
(a,b) \in C(G\cup P) \mbox{ and } s \in A \}$. This approach allows for
defining different concepts of permission
\cite{makinson-torre:2003,boella-torre-icail:2003}\footnote{\cite{stolpe-jal:2010}
offers a different technical treatment, which is however in line with most
intuitions discussed in \cite{makinson-torre:2003,boella-torre-icail:2003}.}:
\begin{description}
\item[Negative permission:] $(a,x)$ is a negative permission w.r.t. $G$ iff
$(a,\neg x)\not\in \mathit{out}_C (G)$; if $x$ is not prohibited by the system
given $a$, then is negatively permitted under those factual conditions $a$.
\item[Static permission:] $(a,x)$ is a static permission w.r.t. $(G,P)$ iff
$(a, x)\in\mathit{out}_C (G\cup \{ (c,d)\})$ for some $(c,d)\in P$; $(a,x)$ is
statically permitted iff it follows from adding a positive permissive norm to
$G$;
\item[Dynamic permission:] \sloppy $(a,x)$ is a dynamic permission w.r.t.
$(G,P)$ iff $(c,\neg d)\in \mathit{out}_C (G, \cup \{ a, \neg x\})$ for some
$(c,d)\in P$; $(a,x)$ is permitted when, given the obligations in $G$, we
cannot prohibit $x$ under the condition $a$ without prohibiting $d$ under
condition $c$ which is however explicitly permitted by the system.
\end{description}
Another concept of permission was proposed in \cite{stolpe-jal:2010} to
specifically capture the idea of exception\footnote{\cite{stolpe-jal:2010}
proposed two definitions. Here, we report on the simpler one.}:
\begin{description} \item[Exemption:] $(a,x)$ is an exemption w.r.t. $(G,P)$
iff $(a, \neg x)\in \mathit{out}_C (G) \backslash \mathit{out}_C (G) - (c,\neg
d)$ for some $(c, d)\in P$; $(a,x)$ is an exception if the code contains the
prohibition of $x$ under condition $a$ which, unless it is removed, it clashes
with an explicit permission in $P$.
\end{description}
 
Since IOL and Defeasible Logic are different logical systems, which were designed for very
different purposes, it is difficult to compare them also in regard to the
problem of permission. Despite any possible connection between the two logics,
which is still an open research problem (Defeasible Logic in fact characterises a
consequence relation falling within cumulative reasoning
\cite{billington:1993}), an immediate comparison shows significant
differences. In particular, formulas in IOL are based on propositional logic
while rules in Defeasible Logic are built using atomic literals, modal literals, and their
negations. A second difference is that the inference machinery of IOL leads to
derive pairs, while the inference output in Defeasible Logic refers to theories and consists
of sets of tagged literals.

However, there are some general similarities in both approaches. 
\begin{itemize}
	\item First, Defeasible Logic, like IOL, models explicit and implicit permissions by
distinguishing in an analogous manner a consequence relation for obligation
and one for permission.

	\item Second, the definition of negative or weak permission in both formalisms is the same. 

	\item Third, although we have not discussed in our approach the notion of
static permission, it seems relatively simple to capture it in Defeasible Logic: indeed, we
may derive that some $p$ is permitted by making an essential use in the
derivation of a rule for explicit permission. The only feature that makes the
difference with respect to Defeasible Logic is that in IOL static permission admits the
principle ``Ought implies Can'', which does not hold for strong permission in
our approach.
\end{itemize}

Similarly, since both approaches distinguish between permissions rebutting
obligations and permissions providing exceptions, we do not see any difficulty
in capturing in Defeasible Logic the intuition behind the concept of exemption, even though
exceptions are more naturally captured in Defeasible Logic by using the superiority relation
between rules. The concept of dynamic permission can be instead expressed in
Defeasible Logic, but in a different way, due to the sceptical character of Defeasible Logic: if we add a
prohibition for some $p$, which is incompatible with any rule for explicit
permission (or even a defeater), then we cannot derive such a prohibition
(unless it is stated to be stronger than any other rule), and so $p$ is
dynamically permitted.

A novelty of our paper is the introduction of the new operator $\odot$ to
express preferences between explicit permissions. A somehow similar idea has
been suggested (though with different purposes) by
\cite{boella-torre-icail:2003}, where a preference relation among pairs (for
obligations and permissions) was introduced. Technically, it is not clear if
that approach can be reframed in our setting. In fact, adopting that option in
Defeasible Logic would not work, as the superiority relation in Defeasible Logic plays a role in the proof
theory only in case of rule conflicts. A clear advantage of the current
proposal is that it adopts a rich formal language where 
\begin{itemize}
\item modal operators can occur in the applicability conditions of rules (the
inputs in IOL are always factual); 
\item we have two ordering types between
permissions: the one expressed by $\odot$ and the one induced by the
superiority relation which applies to defeaters. 
\end{itemize}

To the best of our knowledge, there is no logical system having linear
complexity with analogous normative reasoning capabilities.


\subsection*{Acknowledgements} 
This work is an extended and revised version of the paper presented at Jurix 
2011 \cite{JURIX2011Perm}. We thanks the anonymous referees for their valuable 
comments.

NICTA is funded by the Australian Government as represented by the
Department of Broadband, Communications and the Digital Economy,
the Australian Research Council through the ICT Centre of
Excellence program and the Queensland Government.

\bibliographystyle{plain}
\bibliography{biblioPerm}

\appendix

\section{Appendix}\label{sec:Appendix}

We now prove the properties related to operations that transform a theory $D$
into an equivalent simpler theory $D'$ (we recall that the term ``simpler''
means a theory with a minor number of symbols in it). The transformations
operate on rules either by removing some elements from
some rules, or by deleting rules from a theory. Given the functional nature of the
operations involved, we will refer to the rules in the target theory with the same
names/labels as the rules in the source theory. Thus, given a rule $r\in D$,
we will refer to the rule corresponding to it in $D'$, if it exists, with the
same label, namely $r$.

\setcounter{thm}{18}
\begin{lem}
  Let $D=(\FACTS,R,>)$ be a theory such that $p\in \FACTS \cap \LIT$. Let 
  $D'=(\FACTS ',R',>')$ be the theory obtained from $D$, where
  \begin{align*}
    \FACTS ' = & \FACTS \setminus\set{p}\\
    R' = & \set{r: A(r)\setminus\set{p} \Rightarrow_{\OBL} C(r)!p |
                r\in R, A(r)\cap\tilde{p}=\emptyset} \cup \\
		 & \set{r: A(r)\setminus\set{p} \Rightarrow_{\PERM} C(r) |
                r\in R, A(r)\cap\tilde{p}=\emptyset} \\
    >' = & > \setminus \set{(r,s),(s,r)| r,s\in R, A(r)\cap\tilde{p}\neq\emptyset}.
  \end{align*}
Then $D\equiv D'$. 
\end{lem}
\begin{proof}
The proof is by induction on the length of a derivation $P$.

For the inductive base, we consider all the modal derivations for a
generic literal $q$ in the theory.

\paragraph{$P(1)=+\partial_{\OBL}q$.}{This is possible in two cases: (1)
$\OBL q\in \FACTS$, or (2) $\OBL \non q \not\in \FACTS$, $\neg \OBL q
\not\in \FACTS$ and $\PERM \non q \not\in \FACTS$, and $\exists r\in
R^{\OBL}[q,i]$ that is applicable for $q$ at $i$ at $P(1)$ and every rule for
$\non q$ is either (a) discarded for $\non q$ at $P(1)$, or (b) defeated by a
stronger rule for $q$ applicable at $P(1)$.

For (1), by construction of $D'$, $\OBL q\in \FACTS$ iff $\OBL q\in \FACTS '$,
thus $+\partial_{\OBL}q$ is provable both in $D$ and in $D'$.

For (2), again by construction of $D'$, the modal literals in the clause do
not appear in $\FACTS$ iff they do not appear in $\FACTS '$. Furthermore, an
obligation rule $r\in R^{\OBL}[q,i]$ is applicable for $q$ at $P(1)$ iff
$i=1$, $A(r)\subseteq \FACTS$, and $\non p \not\in A(r)$ since $p \in \FACTS$.
Therefore $A(r)\subseteq \FACTS$ if $A(r)\setminus\set{p}\subseteq \FACTS '$.
This means that if a rule is applicable at $P(1)$ in $D$ then it is applicable
at $P(1)$ in $D'$. In the other direction, suppose that $r$ is applicable in
$D'$, thus in particular $A(r)\subseteq \FACTS '$. In both cases where $r$ has
$p$ in its antecedent or it is not in $D$, we obtain $A(r)\subseteq \FACTS$,
therefore $r$ is applicable in $D$.

Let us now consider a rule attacking $r$, namely a rule $s\in R[\non q,j]$.
For such a rule, we have to analyse cases (a) and (b) above.

(a) A rule $s$ is discarded for $\non q$ at $j$ at $P(1)$ in $D$ iff: (i)
$\exists a_{i}\in A(s)\cap\LIT$ and $a_{i}\notin \FACTS$, or (ii) $\exists
c_{k}\in C(s)$, $k<j$ such that $c_{k}\in \FACTS$, and $s \in R^{\OBL}[c_k,k]$
by condition (2) of a rule being discarded for $+\partial_{\OBL}$.

For (i), we are sure that $a_i \neq p$ since $p\in \FACTS$ by hypothesis. If
$a_{i}=\non p$, then $s$ is discarded in $D$ and, by construction, the rule is
not in $D'$. Hence $R'[\non q]\subseteq R[\non q]$. If $a_{i}\neq\non p$, then
by construction $a_{i}\notin \FACTS$ iff $a_{i}\notin \FACTS '$. For (ii), if
$c_{k}=p$, then the rule is discarded in $D$, the consequent of $s$ is
truncated at $k$ in $D'$, and $\non q$ does not occur in the consequent of $s$
in $D'$, i.e., $s\notin R'[\non q]$. If $c_{k} \neq p$, then the rule is also
discarded in $D'$ since the only difference between $\FACTS$ and $\FACTS '$ is
that $p$ is in $\FACTS$ but not in $\FACTS '$. To summarise, if a rule is
discarded for $\non q$ at $j$ at $P(1)$ in $D$, then the rule is either not in
$D'$, or discarded in $D'$.

For the other direction, in $R'$ there are no rules containing either $p$, or
$\non p$. Hence, if we have $a_{i}\in A(s)$ and $a_{i}\notin \FACTS '$, then
$a_{i}\notin \FACTS$. Similarly for $c_{k}$, if $c_{k}\in \FACTS '$, then
$c_{k}\in \FACTS$. The difference between $D$ and $D'$ is that in $R$ we have
rules with $\non p$ in the antecedent and rules with $p$ preceding $q$ in the
consequent, and these rules are not in $R'$. Since $p\in \FACTS$, all rules in
$R$ for which there is no corresponding rule in $R'$ are discarded in $D$.

(b) The superiority relation is modified in a way that we only remove
instances where one of the rules is discarded in $D$. But only rules that are
not discarded are active in the clauses of the proof conditions where the
superiority relation is involved.

From the discussion above, if a rule is applicable for $q$ at 1 at $P(1)$ in
$D$, then the rule is also applicable in $D'$. If a rule is discarded for
$\non q$ at index 1 at $P(1)$ in $D$, then the rule is not in $D'$, or it is
discarded in $D'$. If a rule $s$ for $\non q$ is applicable in $D$, then there
is an applicable rule $t$ for $q$ stronger than $s$. The rules $s$ and $t$ are
applicable, so they are in $D'$ and $t>'s$. Thus, $D'\vdash+\partial_{\OBL}q$.

Similarly to the other direction, if a rule is applicable in $D'$ then it is
applicable in $D$, and if it is discarded in $D'$ then it is discarded in
$D$. The additional rules in $D$ have either $\non p$ in the antecedent,
or $p$ in the consequent, thus these rules are discarded in $D$, and for
them clause (2.3.1) of $+\partial_{\OBL}$ applies. Therefore if we have a
derivation of length one of $+\partial_{\OBL}q$ in $D'$, then we have
$D\vdash+\partial_{\OBL}q$. }

\paragraph{$P(1)=+\partial_{\PERM}q$.}{The proof is essentially identical to
the inductive base for $+\partial_{\OBL}q$, with some slight syntactical
modifications dictated by the different proof conditions for
$+\partial_{\PERM}$: (1) $\PERM q\in \FACTS$, or (2) $\OBL \non q \not\in
\FACTS$ and $\neg \PERM q \not\in \FACTS$, and $\exists r\in R^{\PERM}[q,1]$
that is applicable for $q$ at $P(1)$ and every obligation rule for
$\non q$ is either (a) discarded for $\non q$ at $P(1)$, or (b) defeated by a
stronger rule for $q$ applicable at $P(1)$. }

\paragraph{$P(1)=-\partial_{\OBL}q$.}{Clearly clauses (1) and (2.1) of $-\partial_{\OBL}$ hold for
$D$ iff they hold for $D'$, given that $\FACTS$ and $\FACTS '$ have the same
modal literals. For clause (2.2), consider a rule $r\in R^{\OBL}[q,1]$. If
$r$ is discarded for $D$ then, as we argued above, it is not in $R'$, or 
it is discarded in $D'$. Also, all rules discarded in $D'$ are discarded in $D$, and
all rules in $R$ for which there is no corresponding rule in $R'$ are
discarded in $D$ as well. As regards clauses (2.3.1)--(2.3.3), we point out
that condition $t\not > s$ between two applicable rules $t$ and $s$ is clearly
unaffected passing from $D$ to $D'$, and the other way around. }

\paragraph{$P(1)=-\partial_{\PERM}q$.}{This case is a mere variant of the
previous one for the negative provability of obligations. }

\vspace{5mm}

\noindent For the inductive step, the property of equivalence between $D$ and
$D'$ is assumed up to the $n$-th step of a generic proof for a given literal
$l$. 

\paragraph{$P(n+1) = +\partial_{\OBL} q$.}{Clauses (1) and (2.1) of
$+\partial_{\OBL}$ follow the same conditions treated in the inductive base
for $+\partial_{\OBL}q$. As regards clause (2.2), if an applicable rule $r \in R^{\OBL}[q,i]$ for $q$ in $D$ exists, then clauses
(1.1)--(1.5) and (2) of Definition~\ref{defn:APPL+pO} are all satisfied. By
inductive hypothesis, we conclude that clauses (1.1)--(1.4) are satisfied
by $r$ in $D'$ as well, and clause (1.5) is satisfied in $D'$ by the inductive
base. For condition (2), the provability of $c_k$ as an obligation in $D'$ is
given by inductive hypothesis; furthermore, $c_k \not\in \FACTS$ or $\non c_k
\in \FACTS$ iff $c_k \not\in \FACTS '$ or $\non c_k \in \FACTS '$ either when
$c_k = \non p$ or $c_k \neq \non p$ since $\FACTS$ and $\FACTS '$ coincide in
both cases (notice that $c_k \neq p$ by hypothesis).

The direction from rule applicability in $D'$ to rule applicability in $D$ is
straightforward. Therefore, a rule $r \in R^{\OBL}[q,i]$ is applicable for $q$
in $D$ iff it is applicable for $q$ in $D'$. Conditions (2.3.1)--(2.3.3) are
treated like cases (a) and (b) for the corresponding inductive base.}

\paragraph{$P(n+1) = +\partial_{\PERM} q$.}{Again, the
inductive base justifies clauses (1), (2.1), (2.3.1), and (2.3.2) of
$+\partial_{\PERM}q$. Clause (2.2) is satisfied by the same reasoning used
in the inductive step of $+\partial_{\OBL}q$ and
by Definition~\ref{defn:APPL+pP}, whose additional clause (3) is trivially
satisfied by inductive hypothesis.
}

\paragraph{$P(n+1) = -\partial_{\OBL} q$.}{Besides conditions (1), (2.1),
(2.3.1), and (2.3.2) -- treated as usual -- it remains to prove that a rule
for $q$ as a permission at $P(n+1)$ is discarded in $D$ iff it is discarded at
$P(n+1)$ in $D'$. To this end, we follow the same analysis carried out in
$P(n+1) = +\partial_{\OBL} q$ for rule applicability, using the inductive base
and hypothesis, and Definition~\ref{defn:DISC+pO+pP}.}

\paragraph{$P(n+1) = -\partial_{\PERM} q$.}{This case is a mere variant
of the previous one for the negative provability as an obligation.}
\end{proof}

\begin{lem}
  Let $D=(F,R,>)$ be a theory such that  $F\cap\LIT=\emptyset$ and 
  $D\vdash+\partial_{\OBL}p$. Let $D'=(F,R',>')$ be the theory obtained from $D$
  where
  \begin{align*}   
    R' & =  \set{r: A(r)\setminus\set{\OBL p} \Rightarrow_{\OBL} C(r) ! \non p \ominus \non p,\\
& A(r)\setminus\set{\OBL p} \Rightarrow_{\OBL} C(r) \ominus p | r\in R, A(r)\cap\widetilde{\OBL p}=\emptyset}\ \cup \\
& \set{r: A(r)\setminus\set{\OBL p} \Rightarrow_{\PERM} C(r) \ominus \non p | r\in R, A(r)\cap\widetilde{\OBL p}=\emptyset}\\
    >' & = > \setminus \set{(r,s),(s,r)| r,s\in R, 
      A(r) \cap \widetilde{\OBL p}\neq\emptyset}.
  \end{align*} 
  Then $D\equiv D'$.  
\end{lem}

\begin{proof}
The proof is by induction on the length of a derivation $P$.

For the inductive base, we consider all possible derivations of length one for
a generic literal $q$ in the theory, where $\OBL p \in \FACTS$.

\paragraph{$P(1)=+\partial_{\OBL}q$.}{From $D$ to $D'$, the structure of the
proof follows the inductive base for $+\partial_{\OBL}q$ of
Lemma~\ref{lem:facts}, where the cases depending on $F$ are trivially
satisfied since $F = F'$, and the other steps are obtained by substituting
$p$ with $\OBL p$ and $\non p$ with $\widetilde{\OBL p}$.

\noindent From $D'$ to $D$, there must exist an applicable rule $r$ proving
$+\partial_{\OBL}q$ at $P(1)$ in $D'$. Then $A(r) \subseteq \FACTS$. By construction of $D'$, the antecedent of $r$ in $D$ is either
the same, or $A(r) \cup \set{\OBL p}$, while the consequent has either $q$ as
the first element, or only $p$ precedes $q$. Since $\OBL p \in \FACTS$ and $p
\not \in \FACTS$, all the combinations of antecedent and consequent denote
applicable rules in $D$.

As already argued, also in this case if a rule $s$ is discarded at $P(1)$ in
$D$, then it is not in $D'$, or it is discarded in $D'$. In particular, all rules in
$R$ for which there is no corresponding rule in $R'$ have either (i) $\neg
\OBL p$, $\OBL \non p$ or $\PERM \non p$ in the antecedent, or (ii) $\non p$
in the consequent. Since $+\partial_{\OBL} p$ holds, clause (1.2) of
Definition~\ref{defn:DISC+pO+pP} and
Proposition~\ref{prop:oblperm} Parts \ref{oblobl} and \ref{oblperm}
make the rules of the form (i) discarded in $D$. Rules of type (ii) are also
discarded in $D$ since $\OBL p \in \FACTS$.}

\paragraph{$P(n+1)=+\partial_{\OBL}q$.}{The proof is essentially identical to
situation $P(n+1)=+\partial_{\OBL}q$ of Lemma~\ref{lem:facts}. Notice that for
rule applicability, clauses $l \in \FACTS$ (condition (1.5)), and $c_k \not\in
\FACTS$ or $\non c_k \in \FACTS$ (condition (2)) are both true in $D$ and
$D'$, since the set of facts is the same and it does not contain non-modal
literals. }

\paragraph{$P(1)=+\partial_{\PERM}q$, $P(n+1)=+\partial_{\PERM}q$.}{For the
inductive base, notice that by construction of $D'$, the antecedent of a rule
$r$ for permission in $D$ is either the same, or $A(r) \cup \set{\OBL p}$,
while the consequent has either $q$ as the first element, or only $\non p$
precedes $q$. However, applicability and refutability of this rule follow the
analysis carried out for $+\partial_{\OBL}$. We treat the inductive step as
usual, using the inductive hypothesis and the fact that $\FACTS = \FACTS '$. }


\vspace{3mm}

\noindent The hypothesis $\FACTS = \FACTS '$ and the structure of $D'$ can
be easily used to prove the inductive base and the inductive step
for proof tags $-\partial_{\OBL}$ and $-\partial_{\PERM}$.

\vspace{3mm}

\noindent Both the inductive base and the inductive step for
$\pm\partial_{X}q$ with $X=\set{\OBL, \PERM}$ in the case where $\OBL p \not
\in \FACTS$ are straightforward. Indeed, even when $\OBL p \in A(r)$ in $D$,
the hypothesis $+\partial_{\OBL}p$ allows an applicable rule in $D'$ to be
also applicable in $D$. The same hypothesis allows us to conclude that a
discarded rule in $D$ is also discarded in $D'$, and the other way around.
\end{proof}

\begin{lem}
  Let $D=(F,R,>)$ be a theory such that  $F\cap\LIT=\emptyset$ and 
  $D \vdash -\partial_{\OBL} p$. Let $D'=(F,R',>')$ be theory obtained from $D$
  where
  \begin{align*}   
    R' & =  \set{r: A(r)\setminus\set{\neg \OBL p} \Rightarrow_{\OBL} C(r) ! p \ominus p | r\in R, A(r)\cap\set{\OBL p}=\emptyset}\\
    >' & = > \setminus \set{(r,s),(s,r)| r,s\in R, 
      A(r) \cap \set{\OBL p}\neq\emptyset}.
  \end{align*} 
  Then $D\equiv D'$.    
\end{lem}

\begin{proof}
Since $D\vdash +\partial_{\OBL}l$ implies $D\vdash -\partial_{\OBL} \non l$ by
Proposition~\ref{prop:oblperm} Part \ref{oblobl}, the modifications on $R'$ and
$>'$ represent a particular case of Lemma~\ref{lem:obl-alg}, where $l = \non
p$. The only difference is that we eliminate from $D$ rules with $\OBL p$ in
the antecedent, and we modify the antecedent of the others eliminating $\neg \OBL
p$ (we recall that condition $-\partial_{\OBL} p$ makes $\neg \OBL p$
defeasibly proved in our framework). In the case that $R^{\OBL}[p] = \emptyset$, no modifications on the consequent
of rules are made since literal $p$ does not appear in any chain by hypothesis.
\end{proof}

\begin{lem}
  Let $D=(F,R,>)$ be a theory such that  $F\cap\LIT=\emptyset$ and 
  $D\vdash+\partial_{\PERM}p$. Let $D'=(F,R',>')$ be theory obtained from $D$
  where
  \begin{align*}   
    R' =  &\set{r: A(r)\setminus\set{\PERM p} \Rightarrow_{\OBL} C(r)! \non p \ominus \non p |
                r\in R, A(r)\cap\widetilde{\PERM p}=\emptyset} \cup \\
		 & \set{r: A(r)\setminus\set{\PERM p} \Rightarrow_{\PERM} C(r) ! p |
                r\in R, A(r)\cap\widetilde{\PERM p}=\emptyset} \\
    >' = & > \setminus \set{(r,s),(s,r)| r,s\in R, 
      A(r)\cap\widetilde{\PERM p}\neq\emptyset}.
  \end{align*} 
  Then $D\equiv D'$.    
\end{lem}

\begin{proof}
The proof is by induction on the length of a derivation $P$; its structure is the same of that for Lemma~\ref{lem:obl-alg}. We
have to take into account the different proof conditions for permission, and
arrange the proof to analyse the inductive bases and steps of each derivation
either when $\PERM p$ is a fact or not.

For the inductive base, we consider all possible derivations of length one for
a generic literal $q$ in the theory, where $\PERM p \in \FACTS$.

\paragraph{$P(1)=+\partial_{\OBL}q$.}{From $D$ to $D'$, the structure of the
proof follows the inductive base for $+\partial_{\OBL}q$ of
Lemma~\ref{lem:facts}, where the cases depending on $F$ are trivially
satisfied since $F = F'$, and the other steps are obtained by substituting
$p$ with $\PERM p$ and $\non p$ with $\widetilde{\PERM p}$.

\noindent From $D'$ to $D$, there must exist an applicable rule $r$ proving
$+\partial_{\OBL}q$ at $P(1)$ in $D'$. Then $A(r) \subseteq \FACTS$. By
construction of $D'$, the antecedent of $r$ in $D$ is either the same, or
$A(r) \cup \set{\PERM p}$, while the consequent has $q$ as the first element.
In this case $p$ cannot precede $q$ since a permission never precedes an
obligation in a reparation chain. Since $\PERM p \in \FACTS$, then all rules
in $D$ corresponding to applicable rules in $D'$ are themselves applicable in
$D$.

As already argued, also in this case if a rule $s$ is discarded for $\non q$
at $P(1)$ in $D$, then either it is not in $D'$, or it is discarded in $D'$.
In particular, all rules in $R$ for which there is no corresponding rule in
$R'$ have either (i) $\neg \PERM p$ or $\OBL \non p$ in the antecedent,
(ii) $\non p$ precedes $\non q$ if $s$ is an obligation rule, (iii) $p$
precedes $\non q$ if $s$ is a permission rule. Since $+\partial_{\PERM} p$
holds, clauses (1.1) and (1.4) of Definition~\ref{defn:DISC+pO+pP}, and
Proposition~\ref{prop:oblperm} Part \ref{permobl} make the rules of the form (i)
discarded in $D$. Moreover, the rules of type (ii) and (iii) are also
discarded since $\PERM p \in \FACTS$, and by Definition~\ref{defn:DISC+pO+pP}
clauses (2) and (3), respectively.}

\paragraph{$P(n+1)=+\partial_{\OBL}q$.}{The proof is essentially identical to
situation $P(n+1)=+\partial_{\OBL}q$ of Lemma~\ref{lem:facts}. Notice that for
rule applicability, clauses $l \in \FACTS$ (condition (1.5)) and $c_k \not\in
\FACTS$ or $\non c_k \in \FACTS$ (condition (2)) are both true in $D$ and
$D'$, since the set of facts is the same and it does not contain non-modal
literals.}

\paragraph{$P(1)=+\partial_{\PERM}q$, $P(n+1)=+\partial_{\PERM}q$.}{For the
inductive base, notice that by construction of $D'$, the antecedent of a rule
$r$ for permission in $D$ is either the same, or $A(r) \cup \set{\PERM p}$,
while the consequent must have $q$ as first element. However, applicability
and refutability of this rule follow the analysis carried out for
$+\partial_{\OBL}$. We treat the inductive step as usual, using the inductive
hypothesis and the fact that $\FACTS = \FACTS '$. In this case, we do not
consider rules of type (iii), i.e., rules for $\PERM \non q$, since only rules
for obligation can attack rules for permission.}


\vspace{3mm}

\noindent The hypothesis $\FACTS = \FACTS '$ and the structure of $D'$ can
be easily used to prove the inductive base and the inductive step
for proof tags $-\partial_{\OBL}$ and $-\partial_{\PERM}$.

\vspace{3mm}

\noindent Both the inductive base and the inductive step for
$\pm\partial_{X}q$ with $X=\set{\OBL, \PERM}$ in the case where $\PERM p \not
\in \FACTS$ are straightforward. Indeed, even when $\PERM p \in A(r)$ in $D$,
the hypothesis $+\partial_{\PERM}p$ allows an applicable rule in $D'$ to be
also applicable in $D$. The same hypothesis allows us to conclude that a
discarded rule in $D$ is also discarded in $D'$, and the other way around.
\end{proof}

\begin{lem}
  Let $D=(F,R,>)$ be a theory such that  $F\cap\LIT=\emptyset$ and 
  $D\vdash-\partial_{\PERM}p$. Let $D'=(F,R',>')$ be theory obtained from $D$
  where
  \begin{align*}   
    R' & =  \set{r: A(r)\setminus\set{\neg \PERM p} \Rightarrow_{\PERM} C(r) \ominus p | r\in R, A(r)\cap\set{\PERM p}=\emptyset}\\
    >' & = > \setminus \set{(r,s),(s,r)| r,s\in R, 
      A(r)\cap\set{\PERM p}\neq\emptyset}.
  \end{align*} 
  Then $D\equiv D'$.    
\end{lem}

\begin{proof}
Condition $D \vdash -\partial_{\PERM}\non l$ is another consequence of $D\vdash
+\partial_{\OBL} l$, as stated by
Proposition~\ref{prop:oblperm} Part \ref{oblperm}. As such, the proof is derived by
Lemma~\ref{lem:obl-alg}, where $l = \non p$.
\end{proof}
\addtocounter{thm}{1}
\begin{lem}
  Let $D=(F,R,>)$ be a theory such that $F\cap\MODLIT=\emptyset$, $\exists r\in 
  R^{\OBL}[p,1]$, $A(r)=\emptyset$, and $R[\non p]\subseteq R_{infd}$. 
  Then $D\vdash+\partial_{\OBL}p$.
\end{lem}

\begin{proof}
Given that there are no modal literals in $F$ clause (2.1) of 
$+\partial_{\OBL}$ is satisfied. Let $r$ be a rule that meets the conditions of 
the Lemma. According to Definition \ref{defn:APPL+pO}, rule $r$ is trivially 
applicable for $p$ in the condition for $+\partial_{\OBL}$, and thus clause (2.2) is applicable 
as well. Finally, for clause (2.3) we have that all rules for $\non p$ are 
inferiorly defeated by an appropriate rule with empty antecedent for $p$, but 
a rule with empty body is applicable. Hence, all clauses for proving 
$+\partial_{\OBL}$ are satisfied. Thus, $D\vdash+\partial_{\OBL}p$.
\end{proof}

\begin{lem}
  Let $D=(F,R,>)$ be a theory such that $F\cap\MODLIT=\emptyset$, $\exists r\in 
  R^{\PERM}[p,1]$, $A(r)=\emptyset$, and $R^{\OBL}[\non p]\subseteq R_{infd}$.
  Then $D\vdash+\partial_{\PERM}p$.
\end{lem}
\begin{proof}
The proof is analogous to the previous one. The differences are that we have to 
use the notion of applicability in Definition \ref{defn:APPL+pP}, and that the 
rules that are inferiorly defeated are restricted to rules in $R^{\OBL}[\non 
p]$.
\end{proof}

\begin{lemma}
  Let $D=(F,R,>)$ be a theory such that $F\cap\MODLIT=\emptyset$ and 
  $R^{X}[p]=\emptyset$, for $X\in\set{\OBL,\PERM}$. Then $D\vdash-\partial_Xp$.
\end{lemma}
\begin{proof}
If there are no modal literals and the set of defeasible (obigation/permission) 
rules for a literal $p$ is empty, then clause (2.2) of
$-\partial_{\OBL}$ and $-\partial_{\PERM}$ are vacuously satisfied.
\end{proof}

\begin{lemma}
  Let $D=(F,R,>)$ be a theory such that $F\cap\MODLIT=\emptyset$, and $\exists 
  r\in R[p,1]$ such that $A(r)=\emptyset$ and $r_{sup}=\emptyset$. Then
  \begin{enumerate}
    \item if $r\in R^{\OBL}$, then $D\vdash-\partial_{\Box}\non p$, 
      $\Box\in\set{\OBL,\PERM}$;
    \item if $r\in R^{\PERM}\cup R_{def}$, then $D\vdash-\partial_{\OBL}\non p$. 
  \end{enumerate}
\end{lemma}
\begin{proof}
Let $r$ be a rule in a theory $D$ for which the conditions of the Lemma hold. 
It is easy to verify that for both cases the rule satisfies clause (2.3) of 
$-\partial_{\Box}$, in particular (2.3.2--3) for $-\partial_{\OBL}$ and (2.3.2) 
for $-\partial_{\PERM}$.
\end{proof}

\end{document}